\documentclass[12pt,a4paper,oneside]{book}
\newcommand{\rand}{\stackrel{\$}\leftarrow}
\newcommand{\api}{\stackrel{\textit{def}}=}
\usepackage{textcomp}

\usepackage{multicol}
\usepackage{nccmath}
\usepackage{amsfonts}
\usepackage[utf8]{inputenc}
\usepackage[english]{babel}
\usepackage{amsmath,amssymb}
\usepackage{longtable}
\usepackage{graphicx}

\long\def\/*#1*/{}

\usepackage{hyperref}

\usepackage[top=2.25in, bottom=.25in, left=1.25in, right=1.25in]{geometry}
\usepackage[table]{xcolor}
\usepackage{fancyhdr}
\usepackage[T1]{fontenc}

\renewcommand{\arraystretch}{1.5}
\newcolumntype{P}[1]{>{\centering\arraybackslash}p{#1}}
\newcolumntype{M}[1]{>{\centering\arraybackslash}m{#1}}
\usepackage{booktabs}

\usepackage{amsthm}
\usepackage{amssymb}
\usepackage{adjustbox}

\theoremstyle{definition}
\newtheorem{theorem}{Theorem}[section]

\newtheorem{lemma}{Lemma}
\newtheorem{claim}{Claim}

\newcommand{\la}{\leftarrow}
\newcommand{\ra}{\rightarrow}
\newcommand{\nprp}{\mbox{$\widetilde{\mbox{\rm\small prp}}$}}

\usepackage[many]{tcolorbox}

\newcommand{\nbits}{\{0,1\}^n}
\newcommand{\mbits}{\{0,1\}^*}
\newcommand{\tbits}{\{0,1\}^{32}}
\usepackage{titlesec} 
\usepackage{lipsum} 

\usepackage{csquotes}

\usepackage{algorithm}
\usepackage{algorithmic}

\hypersetup{
colorlinks,
citecolor=black,
filecolor=black,
linkcolor=black,
urlcolor=black
}

\makeatletter 
\def\cleardoublepage{\clearpage\if@twoside \ifodd\c@page\else%
    \hbox{}%
    \thispagestyle{empty}
    \if@twocolumn\hbox{}\newpage\fi\fi\fi} 
\makeatother

\begin{document}
\pagenumbering{roman}
\begin{titlepage}
\begin {center}
\vspace*{0.3cm}

\centering
{\huge\bf 
\begin{center}
Security of XCB and HCTR
\end{center}
}
~\\
\vspace {0.35in}
\textbf{\textsc{\large Dissertation Submitted In Partial Fulfilment Of The Requirements For The Degree Of}}\\
\sf
\vspace {0.3in}
{\bf Master of Technology}\\
{\bf in}\\
{\bf Computer Science}\\
\vspace {.3in}
\large
{by}\\
\vspace{.2in}
{\Large \bf Manish Kumar} \\
{\small [ Roll No: CS-1616 ]} \\
\vspace {.4in}
{Under the Guidance of}\\
\vspace{.3in}
{\Large \bf Dr. Debrup Chakraborty} \\
{Associate Professor \& Head}\\
{Cryptology and Security Research Unit (CSRU)}\\
\vspace {0.2in}

\begin{figure}[htbp]
{\centering \resizebox*{!}{3.5cm}{\includegraphics{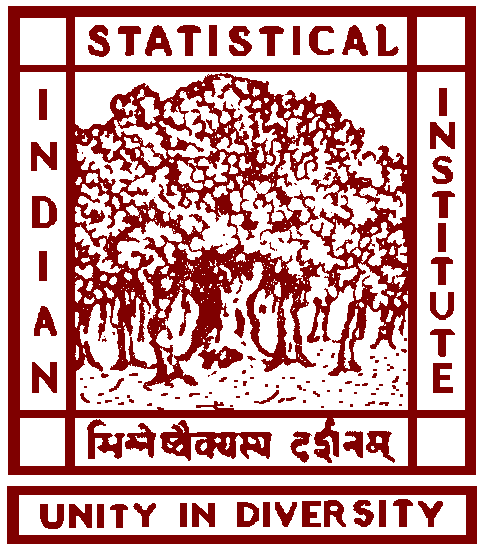}} \par}
\end{figure}
\vspace {0.2in}
{
{\bf \large Indian Statistical Institute}\\
{\bf Kolkata-700108, India}\\
\vspace {0.3in}

{\bf July 2018}
}
\end {center}
\end{titlepage}

\chapter*{}
\thispagestyle{empty}
\begin{center}
\vspace*{-.6in}
\parskip=0.1in
{\Large \bf Declaration} 
\end{center}
\vspace{0.4in}
\hspace{0.5in}


%

\noindent
I hereby declare that the dissertation 
report entitled {\bf ``Security of XCB and HCTR''} submitted to Indian Statistical Institute, Kolkata, is a \textit{bona fide} record of work carried out in partial fulfilment for the award of the degree of {\bf Master of Technology in Computer Science}. The work has been carried out under the guidance of \textbf{Dr. Debrup Chakraborty}, Associate Professor, CSRU, Indian Statistical Institute, Kolkata. \\

\vspace{2.5mm}
\noindent I further declare that this work is original, composed by myself. The work contained herein is my own except where stated otherwise by reference or acknowledgement, and that this work has not been submitted to any other institution for award of any other degree or professional qualification.

\vspace{1.7in}
\noindent
\begin{minipage}[t]{5cm}
\begin{flushleft}
Place : Kolkata\\
Date : \today\\
\end{flushleft}
\end{minipage}
\hfill
\begin{minipage}[t]{9cm}
\begin{flushright}
{\bf Manish Kumar}\\
Roll No: CS-1616\\
Indian Statistical Institute\\
Kolkata - 700108 , India. 
\end{flushright}
\end{minipage}

\newpage

\chapter*{}
{\thispagestyle{empty}
\begin{center}
\vspace*{-.6in}
{\Large \bf CERTIFICATE} 
\end{center}
\vspace{0.4in}
\hspace{0.5in}

{\noindent
This is to certify that the dissertation entitled {\bf ``Security of XCB and HCTR''}
submitted by {\bf Manish Kumar} to Indian Statistical Institute, Kolkata, 
in partial fulfilment for the award of the degree of 
{\bf Master of Technology in Computer Science} is a \textit{bona fide} record of work carried out by 
him under my supervision and guidance. The dissertation has fulfilled all the 
requirements as per the regulations of this institute and, in my opinion, has 
reached the standard needed for submission.\\\\\\\\\\\\

\begin{tabbing}
a\=aa\=aaaa\=aaaaaaaaaaaaaaaaaaaaaaaaaaaaaaaaaaaa\= \kill
\>\rule[.2em]{25em}{0.5pt}~\\
\>{\large \bf Debrup Chakraborty }\\
\>Associate Professor \& Head,\\
\>Cryptology and Security Research Unit,\\
\>Indian Statistical Institute,\\
\>Kolkata-700108, India.\\
\end {tabbing}
}
}

\newpage

\thispagestyle{empty}
\vspace*{5.5cm}

\begin{flushleft}
{ \em ``There exist attackers who follow non-violence.''}
\end{flushleft}

\newpage

\chapter*{}
\thispagestyle{empty}
\begin{center}
\vspace*{-.6in}
\parskip=0.1in
{\Large \bf Acknowledgements} 
\end{center}
\vspace{0.4in}
\hspace{0.5in}

\noindent
I would like to take this opportunity to thank people who are behind my success in this project.

\vspace{2.5mm}
\noindent
\textit{Prima facie}, I would like to thank my parents, family members and teachers who supported me in every walk of my life.

\vspace{2.5mm}
\noindent
I would like to show my highest gratitude to my adviser, \textit{Prof. Debrup Chakraborty} of Cryptology and Security Research Unit, for his guidance
and continuous support and encouragement. His zeal and method of teaching are highly motivating.

\vspace{2.5mm}
\noindent
I would also like to thank \textit{A V S D Bharadwaj}, student research group of CSRU and ``M.Tech. Crypto Cluster'' for their valuable suggestions and discussions. 

\vspace{2.5mm}
\noindent
My deepest thanks to all the professor of Indian Statistical Institute, for their valuable suggestions which added an important dimension to my research work.

\vspace{2.5mm}
\noindent
Last but not the least, I would like to thank all of my friends for their help. I would also like to thank all those, 
whom I have missed out from the above list.

\vspace{.3in}
\begin{flushright}
{\bf Manish Kumar}\\
Indian Statistical Institute\\
Kolkata - 700108 , India. 
\end{flushright}

\thispagestyle{empty}
\vspace*{5.5cm}
\begin{flushright}

{\Large \em To my family and supervisor}

\end{flushright}

\newpage

\chapter*{}
\vspace*{-.6in}
\begin{center}
\Large \bf {Abstract}
\end{center}
\vspace{0.4in}

\par Tweakable Enciphering Scheme (TES) is a length preserving scheme which provides confidentiality and admissible integrity. XCB (Extended Code Book) is a TES which was introduced in 2004. In 2007, it was modified and security bound was provided. Later, these two versions were referred to as XCBv1 and XCBv2 respectively. XCBv2 was proposed as the IEEE-std 1619.2 2010 for encryption of sector oriented storage media. In 2013, first time Security bound of XCBv1 was given and XCBv2's security bound was enhanced. A constant of $2^{22}$ appears in the security bounds of the XCBv1 and XCBv2.\\ 

\par We showed that this constant of $2^{22}$ can be reduced to $2^{5}$. Further, we modified the XCB (MXCB) scheme such that it gives better security bound compared to the present XCB scheme. We also analysed some weak keys attack on XCB and a type of TES known as HCTR (proposed in 2005). We performed distinguishing attack and the hash key recovery attack on HCTR. Next we analysed the dependency of the two different keys in HCTR.\\

\par \textbf{Keywords:} Disk encryption $\cdot$ IEEE-std 1619.2 2010 $\cdot$ Tweakable enciphering scheme $\cdot$ XCB $\cdot$ MXCB $\cdot$ Weak keys $\cdot$ HCTR.

\newpage
\tableofcontents
\newpage

\listoffigures

\listoftables

\mainmatter
	
\chapter{Introduction}

\par Day-by-day we are becoming more reliable on the data which can be personal, organisational, top secret, anything or everything. Often we come across the news of data breach and  fraud. Recently (in 2018), Cambridge Analytica was in news due to data breach of facebook users. In India, Unique Identification Authority of India (UIDAI) which is responsible for AADHAR scheme is facing data breach on daily basis. In 2012, we generated approx 2.5 exabytes of data every day. Majority ($90 \%$) of all the present data is generated in last few years. As data is asset of the $21^{st}$ century and it is on continuous threat, question arises whether we can make the data secure or not? If yes then up-to what extent? As a wise man said, ``The hardest thing of all is to find a black cat in a dark room, especially if there is no cat''. Providing security is more or less like that.

\par At any instant, data can be in two states - stored or in transit.
Generally, we are interested in the confidentiality and integrity of the data in both of these states and the techniques involved in these two different scenarios are different. In most schemes which provide integrity/confidentiality in a strong sense
a length expansion of the original data takes place, i.e., the transformed data occupies more space than the original data. Though such expansion 
can be easily tolerated for most scenarios, there are specific application areas where such length expansion cannot be tolerated. One such area is the application of low level disk encryption.  
Disk encryption ensures confidentiality and integrity of the stored data. Even if the hardware containing the disk is stolen the data stored in
it would be unreadable to the adversary, also the adversary would be unable to change the contents of the disk in a meaningful way.

In low level disk encryption, the encryption/decryption algorithm resides on the disk controller and sees the disk as a bare collection 
of sectors. It encrypts the data before storing it into the sector and decrypts it after reading a sector and before sending it to high level 
applications. As each sector is of fixed length (4096 bytes, in modern disks), so length expansion after encryption cannot be tolerated.
A well accepted solution for the problem of disk encryption is a cryptographic object called a Tweakable Enciphering Scheme (TES).

\par A TES (Wide block encryption) is a block-cipher mode of operation with the following properties:
\begin{enumerate}
	\item {\em Length preserving:} The length of the cipher text is same as the length of the plain text.
	\item {\em Ciphertext Variability:} A TES takes as input a key, a plaintext and a special quantity called the
	{\em tweak}. If the same message is encrypted with the same key but different tweaks then unrelated
	ciphertexts are obtained. In the application of disk encryption, sector addresses are considered as tweaks. This
	property ensures that even if the same data is stored in two different sectors the cipher texts would look different.
	\item {\em Confidentiality:} The cipher texts produced by a TES are indistinguishable from random strings to any  computationally bounded adversary.
	Which means that any practical adversary would see the cipher texts as random strings and would be thus unable to derive any
	information regarding the plain text which produced this cipher text.
	\item {\em Integrity:} If a single bit of a valid ciphertext produced by a TES is changed, then this altered ciphertext on decryption will produce
	a random looking plaintext. This property ensures that an adversary would not be able to alter a ciphertext so that it gets decrypted to something meaningful.
\end{enumerate}

In designing a TES, we have certain goals. We can list them as follows:
\begin{enumerate}
	\item First and foremost goal is TES's correctness i.e. decryption should undo encryption for every message in the message space.
	\item Our next aim, TES should be as efficient as possible. Some of the important dimensions of efficiency expected of a TES are the following:
	\begin{itemize}
		\item Less running time for encryption/decryption.
		\item Small circuit area when implemented in hardware.
		\item Low power consumption when implemented in power constrained devices.
		\item Small code size and low memory usage in memory constrained software implementations.
	\end{itemize} 
	\item In TES, we are using tweak for getting variability in the output for the same input. So to change the tweak should be cheaper than changing the key. In most of the block cipher changing the encryption key is relatively expensive since there is need to perform ``key setup'' operation.
	\item TES should be secure i.e. even the adversary has control of the tweak input than also the scheme should be secure.
\end{enumerate}

Designing efficient TESs which provide the required security in provable terms is a challenging problem.
In the last two decades there has been some intense work in designing and proving the security of TES.
Some of the existing constructions are PEP \cite{chakraborty2006new}, HCTR\cite{wang2005hctr}, HCH\cite{chakraborty2006hch}, TET \cite{halevi2007invertible}, HEH \cite{sarkar2009efficient}, CMC\cite{halevi2003tweakable}, XCB\cite{mcgrew2004extended,mcgrew2007security} and EME\cite{halevi2004parallelizable}.
TES has been standardised because of its practical application in disk encryption.

\section{XCB and HCTR}
In this dissertation we study two TES called XCB and HCTR.

\par XCB (Extended Code Book) and HCTR are Hash-Counter-Hash Tweakable Enciphering Schemes i.e. both the schemes use first a layer of universal hash function then a counter (CTR) mode and then again another layer of universal hash function for encryption and decryption. The universal hash functions used in both constructions  is a variant of polynomial evaluation hash \cite{wegman1981new}. Let us consider a message $M$ which is parsed into $m$ blocks $M_1,M_2,\ldots,M_m$, each of length $n$ bits. To hash $M$ using a polynomial hash with
a key $n$ bit key $H$, the polynomial $h_H(M) = \sum_i M_iH^i$ is computed. Variants of this  polynomial evaluation hash  has been widely used to construct message authentication codes (MAC) \cite{brassard1983computationally,mcgrew2004security,wegman1981new}, authenticated encryption (AE), TES \cite{mcgrew2004extended,chakraborty2006hch,wang2005hctr} and other cryptographic schemes. CTR mode uses the block cipher to generate the key stream used in the message encryption: $E_K(S_i), i = 1, 2, \cdot \cdot \cdot $, where $K$ is the key of block cipher and $S_i$ is the number generated by a counter. The main difference between HCTR and XCB is: HCTR has two master keys, one for counter mode and other for universal hash function while XCB has only one master key from that other keys are generated. XCB and HCTR use different variants of the Counter mode and the polynomial evaluation hash \cite{sun2015weak}.\\

\par In 2004, McGraw and Fluherer proposed Tweakable Enciphering Scheme (TES) named as XCB in \cite{mcgrew2004extended} without providing a proof. Later in 2007, they made changes in original construction and proved security of the updated construction in \cite{mcgrew2007security}. Authors claim that the changes were made for the improvement of performance of XCB and make it easier to analyse. Later Chakraborty, Hernandez-Jimenez and Sarkar \cite{chakraborty} did a detailed analysis of two versions of XCB as described in \cite{mcgrew2004extended} and \cite{mcgrew2007security}. The study in \cite{chakraborty} names the version of XCB in \cite{mcgrew2004extended} as XCBv1 and the one in \cite{mcgrew2007security} as XCBv2, we will also follow the same nomenclature. The analysis in \cite{chakraborty} concludes that the security claims regarding XCBv2 as presented in \cite{mcgrew2007security} are largely erroneous. XCBv2 is completely insecure for certain types of messages, in particular, there is an easy distinguishing attack on XCBv2 if it is used on messages whose length is not a multiple of the block length $n$ of the underlying block cipher. Though XCBv2 is secure for other messages, the proof and the security bound was shown to be incorrect. In \cite{chakraborty} a correct security bound for XCBv2 (message for which it is secure) was derived and also a proof for XCBv1 was provided. That proof was based on the analyses done in \cite{iwata2012breaking}. In \cite{iwata2012breaking} the security bounds of an authenticated encryption scheme called GCM were analyzed, as GCM and XCB shares almost the same hash function. Hence, the techniques used and analyze GCM in \cite{iwata2012breaking} could be adopted to analyze XCB in \cite{chakraborty}. Further, analysis of the GCM bound was done in \cite{Iwata}. We use the analyses done in \cite{Iwata} to give an improved security bound on XCBv1 and XCBv2. We also modify the XCB (MXCB say) and give its security bound. Further, we compare the improved security bound and MXCB with some existing TES mode having parameter of practical value followed by some weak keys analysis on XCB.

We also do some analysis on HCTR. It was proposed by Wang, Feng and Wu in 2005. It is a mode of operation which provides a tweakable strong pseudorandom permutation \cite{wang2005hctr}. We show how the hash function is insecure. We perform distinguishing attack and the hash key recovery attack on HCTR. Next we analyse the dependency of the two different keys in HCTR. In particular, we analyse the following scenario. Suppose HCTR with keys $K$ and $h$ has been used for some time, and $K$ gets compromised. We show that only changing $K$ would rise to a completely insecure scheme.

\section{Outline of dissertation} In \autoref{chap : Preliminareis}, we discuss and formalise the notion of security for Tweakable Enciphering Scheme. In \autoref{chap : XCB}, we formalise the XCB and prove a lemma. Then gives the security proof of XCBv1 followed by security bound for MXCB and comparison with some existing TES. Also, weak keys analysis on XCB. In \autoref{chap : HCTR}, we discuss the construction of HCTR, distinguishing and key recovery attack on the existing HCTR scheme followed by key dependency of the master keys . In \autoref{chap : Conclusion}, we conclude the discussion.
\chapter{Preliminaries}\label{chap : Preliminareis}

\par Following are the notation which we will use in subsequent chapter.
\section{Notation}
The set of all $n$-bit strings will be denoted by $\{0, 1\}^n$. For a binary string $X$, $|X|$ will represent the size of the string in bits. We will use $X\|Y$ for concatenating binary string $X$ and $Y$; for $r \leq |X|,$ $r$ left most and $r$ right most bits of $X$ would be denoted by $\text{msb}_r(X)$ and $\text{lsb}_r(X)$ respectively. By int($X$) we denote the integer represented by the binary string $X$, ${\sf bin}_n(i)$ will denote the $n$-bit binary representation of $i$, where the leftmost bit is the most significant bit and $i$ is non-negative such that $i \leq 2^n-1$. For $X, Y \in \{0, 1\}^n, X \oplus Y$ and $XY$ will respectively denote addition and multiplication in $GF (2^n)$. We denote $parse_n(X)$ by $(X_1, X_2,\dotsc,X_m)$ where each $X_i$ is of $n$-bit except last one while $1\leq|X_m|\leq n$ and cardinality of $X$ would be denoted by $\#X.$ 

\par In standard of XCB, field $GF (2^{128})$ is represented by the irreducible polynomial $x^{128} + x^7 + x^2 + x + 1$. Note that selection of irreducible polynomial doesn't affect the security of scheme. Therefore, proofs and attacks are irrespective of the irreducible polynomial.

\section{Tweakable Enciphering Schemes (TES)} A Tweakable Enciphering Scheme is a pair of functions $(\textbf{E, D})$ where $\textbf{E}$ and $\textbf{D}$ are the encryption and decryption functions respectively of the enciphering scheme. Here, encryption $\textbf{E} : \mathcal{K} \times \mathcal{T} \times \mathcal{M} \ra  \mathcal{M}$ and decryption $\textbf{D} : \mathcal{K} \times \mathcal{T} \times \mathcal{M} \ra  \mathcal{M},$ where $\mathcal{K}$ and $\mathcal{T}$ are non-empty sets, and they denote the key space and the tweak space respectively. The message and the cipher space $\mathcal{M} \subseteq \bigcup_{i\ge 1}\{0, 1\}^i.$
We will denote $\textbf{E}(K, T,\cdot)$ by $\textbf{E}^T_K(\cdot)$ and $\textbf{D}(K, T,\cdot)$ by $\textbf{D}^T_K(\cdot).$
Encryption and decryption are length preserving i.e. for every $K \in \mathcal{K}$, $M \in \mathcal{M}$ and $T \in \mathcal{T}$ such that $|\textbf{E}_K^T(X)| = |X|.$ For the correction purpose, $X = \textbf{D}^T_K(Y)$ if and only if  $\textbf{E}^T_K(X) = Y$ where $\textbf{D} = \textbf{E}^{-1}.$


\section{Security of TES}
Discussion of this section is based on \cite{halevi2004parallelizable}. An $n$-bit block-cipher is a function $E:\mathcal{K} \times \nbits \rightarrow \nbits$, where $\mathcal{K} \neq \phi$ is the key space for any key $K \in \mathcal{K}$ and $E(K,\cdot)$ is a permutation. 

\par An adversary $A$ is a probabilistic algorithm which can access two oracles and gives output either 0 or 1. The notation $A^{\mathcal{O}_1,\mathcal{O}_2}\Rightarrow 1$ denotes the events that the adversary $A,$ interacts with the oracles $\mathcal{O}_1, \mathcal{O}_2$ and finally output the bit 1. Event of choosing $X$ uniformly at random from the finite set $S$ is represented by $X \rand S$.
\par Let Perm($n$) denotes the set of all permutation on $\nbits$. The advantage of $A$ in breaking the strong pseudo randomness  of $E$ is defined as
\begin{multline*}
\textbf{Adv}^{\pm prp}_E (A) = \bigg|Pr\bigg[K \rand \mathcal{K} : A^{E_K(\cdot),E_K^{-1}(\cdot)}\Rightarrow 1\bigg]\\ - Pr\bigg[\pi \rand \text{Perm}(n): A^{\pi(\cdot), \pi^{-1}(\cdot)}\Rightarrow 1\bigg]\bigg|. 
\end{multline*}

\par Let $Perm^{\mathcal{T}}(\mathcal{M})$ denote the set of all functions $\boldsymbol{\pi}: \mathcal{T} \times \mathcal{M} \rightarrow \mathcal{M}$ 
where $\boldsymbol{\pi}(T,\cdot)$ is a length preserving  permutation on $\mathcal{M}$ and $\boldsymbol{\pi} \in Perm^{\mathcal{T}}(\mathcal{M})$ is known as indexed permutation. For a Tweakable Encipher Scheme $\textbf{E}: \mathcal{K} \times \mathcal{T} \times \mathcal{M} \rightarrow \mathcal{M}$, we define the advantage an adversary $A$ in distinguishing $\textbf{E}$ and its inverse from a random tweak indexed permutation and its inverse  in the following way:
\begin{multline}\label{sprp}
\textbf{Adv}^{\pm \widetilde{prp}}_\textbf{E} (A) = \bigg|Pr\bigg[K \rand \mathcal{K} : A^{\textbf{E}_K(\cdot, \cdot),\textbf{E}_K^{-1}(\cdot,\cdot)}\Rightarrow 1\bigg]\\ - Pr\bigg[\boldsymbol{\pi} \rand Perm^{\mathcal{T}}(\mathcal{M}): A^{\boldsymbol{\pi}(\cdot, \cdot), \pi^{-1}(\cdot, \cdot)}\Rightarrow 1\bigg]\bigg|. 
\end{multline}

\par 
We define $\textbf{Adv}^{\pm \widetilde{prp}}_\textbf{E} (q, \sigma_n)$
by $max_A\textbf{Adv}^{\pm \widetilde{prp}}_\textbf{E} (A)$ where maximum is taken over all adversaries which makes at most $q$ queries having at most $\sigma_n$ many blocks. For a computational advantage we define $\textbf{Adv}^{\pm \widetilde{prp}}_\textbf{E} (q, \sigma_n, t)$ by $max_A\textbf{Adv}^{\pm \widetilde{prp}}_\textbf{E} (A)$. In
addition to the previous restrictions on $A$, he can run in time at most $t$.

\chapter{Security of XCB}\label{chap : XCB}

\par In 2004, McGraw and Fluhrer proposed Tweakable Enciphering Scheme (TES) named as XCB in \cite{mcgrew2004extended} without providing a proof. Later in 2007, they made changes in original construction and proved security of the updated construction in \cite{mcgrew2007security}. Authors claim that the changes were made for the improvement of performance of XCB and make it easier to analyse. Later Chakraborty, Hernandez-Jimenez and Sarkar \cite{chakraborty} did a detailed analysis of two versions of XCB as described in \cite{mcgrew2004extended} and \cite{mcgrew2007security}. The study in \cite{chakraborty} names the version of XCB in \cite{mcgrew2004extended} as XCBv1 and the one in \cite{mcgrew2007security} as XCBv2, we will also follow the same nomenclature. The analysis in \cite{chakraborty} concludes that the security claims regarding XCBv2 as presented in \cite{mcgrew2007security} are largely erroneous. XCBv2 is completely insecure for certain types of messages, in particular, there is an easy distinguishing attack on XCBv2 if it is used on messages whose length is not a multiple of the block length $n$ of the underlying block cipher. Though XCBv2 is secure for other messages, the proof and the security bound was shown to be incorrect. In \cite{chakraborty} a correct security bound for XCBv2 (message for which it is secure) was derived and also a proof for XCBv1 was provided. That proof was based on the analyses done in \cite{iwata2012breaking}. In \cite{iwata2012breaking} the security bounds of an authenticated encryption scheme called GCM were analyzed, as GCM and XCB shares almost the same hash function. Hence, the techniques used and analyze GCM in \cite{iwata2012breaking} could be adopted to analyze XCB in \cite{chakraborty}. Further, analysis of the GCM bound was done in \cite{Iwata}. We use the analyses done in \cite{Iwata} to give an improved security bound on XCBv1 and XCBv2.

\par In this chapter, we give the improved security bounds on XCBv1 and XCBv2 (with full block). Further, we modify the XCB (MXCB say) and give its security bound. Also, we compare the improved security bound and MXCB with some existing TES mode having parameter of practical value. After that in last section of the chapter, we show some weak keys analysis on XCB. 
\section{Description of XCB}
\begin{figure}[!ht]
	\includegraphics[width=0.85\textwidth]{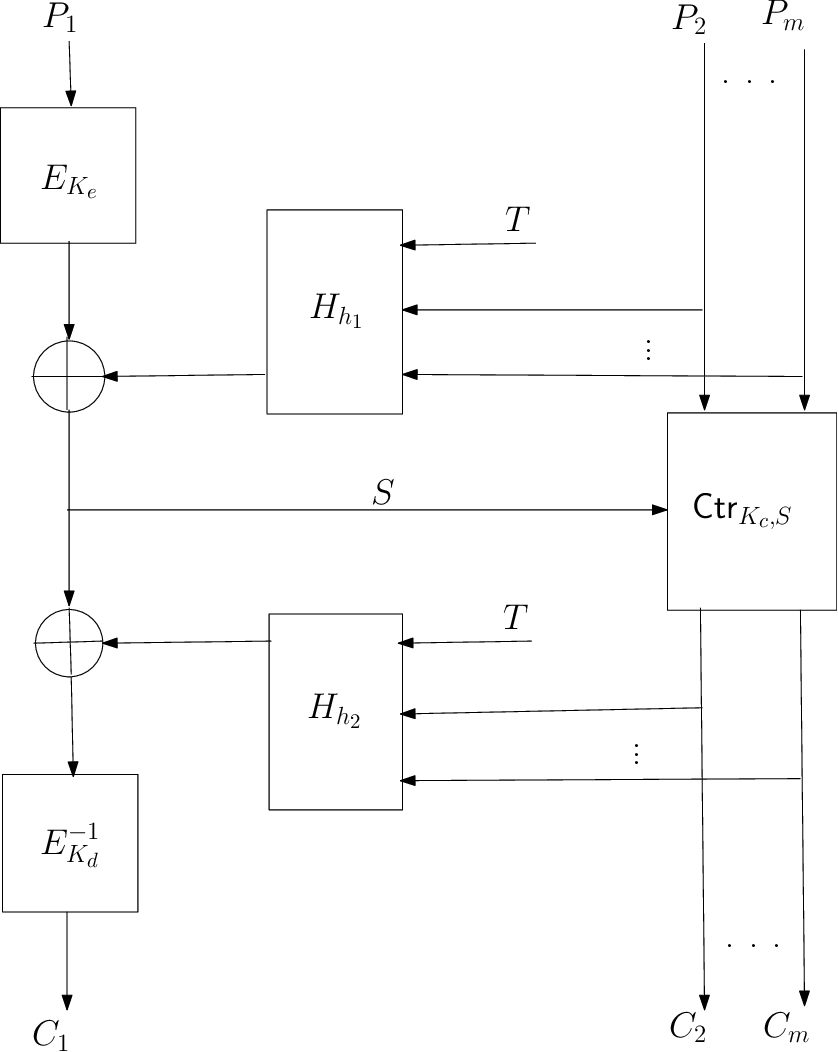}
	\centering
	\caption{\label{cons:xcbv1} Encryption of XCBv1}
\end{figure}

XCB is hash-counter-hash scheme which use hash function and counter mode as the basic building blocks of the scheme. The construction of XCBv1 and XCBv2 are shown in the \autoref{cons:xcbv1} and \autoref{cons:xcbv2} respectively, and encryption algorithms for XCBv1 and XCBv2 are shown in the \autoref{fig : XCBv1_v2}.\\

Following definition of hash function and {\sf Ctr} mode are as defined in \cite{chakraborty}.
\par \textbf{Hash function} $H : \nbits \times \mathcal{X} \times \mathcal{Y} \rightarrow \nbits $, where $\mathcal{X}, \mathcal{Y}$ are non-empty subsets of $\mbits.$ For $ T \in \mathcal{Y}$ and $ X \in \mathcal{X} $, we write $ H_h(X, T)$ instead of $H(h, X, T)$. The hash function $H$ is defined as

\begin{equation*} 
H_h(X, T) = X_1 h^{m+p+1} \oplus X_2 h^{m+p} \oplus \dotso \oplus  {\sf pad}(X_m)h^{p+2} \oplus T_1 h^{p+1} 
\end{equation*}\begin{equation}\label{1}
\quad \quad \quad \hspace{2.2cm}	\oplus T_2 h^p \oplus \dotso \oplus \text{\sf pad}(T_p)h^2 \oplus (\text{\sf bin}_{\frac{n}{2}}(|X|)\|\text{\sf bin}_{\frac{n}{2}}(|T|))h, 
\end{equation}
where $(X_1, X_2,\dotsc,X_m) = parse_n(X), (T_1, T_2,\dotsc, T_p) = parse_n(T)$. The ${\sf pad}$ function is defined as ${\sf pad}(X_m) = X_m \|0^r$ where $r = n -|X_m|$. Thus, $|{\sf pad}(X_m)| = n$.

\textbf{Counter mode:} ${\sf Ctr}$ with key $K$, counter value $S$ and message $A_1, \dots, A_m$ is defined as:

\centerline{$\text{{\sf Ctr}}_{K,S}(A_1, \dotsc, A_m) = (A_1 \oplus E_K(\text{inc}^0(S)),\dotsc, A_m \oplus E_K(\text{inc}^{m-1}(S)).$}

\par If the last block $A_m$ is incomplete then the quantity $A_m \oplus E_K(\text{inc}^{m-1}(S))$ is replaced by the quantity $A_m \oplus \text{\sf drop}_r(E_K(\text{inc}^{m-1}(S))$ where $r = n -|A_m|$ and $\text{\sf drop}_r(E_K(\text{inc}^{m-1}(S))$ is the first $(n-r)$ bits of $E_K(\text{inc}^{m-1}(S)$. In the definition of {\sf Ctr}, for a bit string $X \in \nbits, \text{inc}(X)$ treats the last significant 32 bits of $X$ as a non-negative integer and increment this value modulo $2^{32}$ i.e.\\
\centerline{$\text{inc}(X) = \text{msb}_{n-32}(X)\| \text{\sf bin}_{32}(\text{int} (\text{lsb}_{32}(X)) + 1 \mod2^{32})$.}

For $r \geq 0$, we write $\text{inc}^r(X)$  to denote the $r$ times iterative applications of \text{inc} on $X$. We  use the convention that $\text{inc}^0(X) = X$. For this specific structure of \text{inc} both XCBv1 and XCBv2 can only be used with block ciphers where the block length $n \geq 32,$ which does not amount to a practical constraint.

\par XCBv1 and XCBv2 has length constraints with respect to plaintext and tweak as $n \le |P| \le 2^{39}$ and $0 \le |T| \le 2^{39}$ respectively. XCBv2 is specified in the standard IEEE 1619.2 2010. In the standard, AES is fixed as the block cipher and message length is always multiple of 8 bits. \\

\begin{figure}[!ht]
	\includegraphics[width=0.75\textwidth]{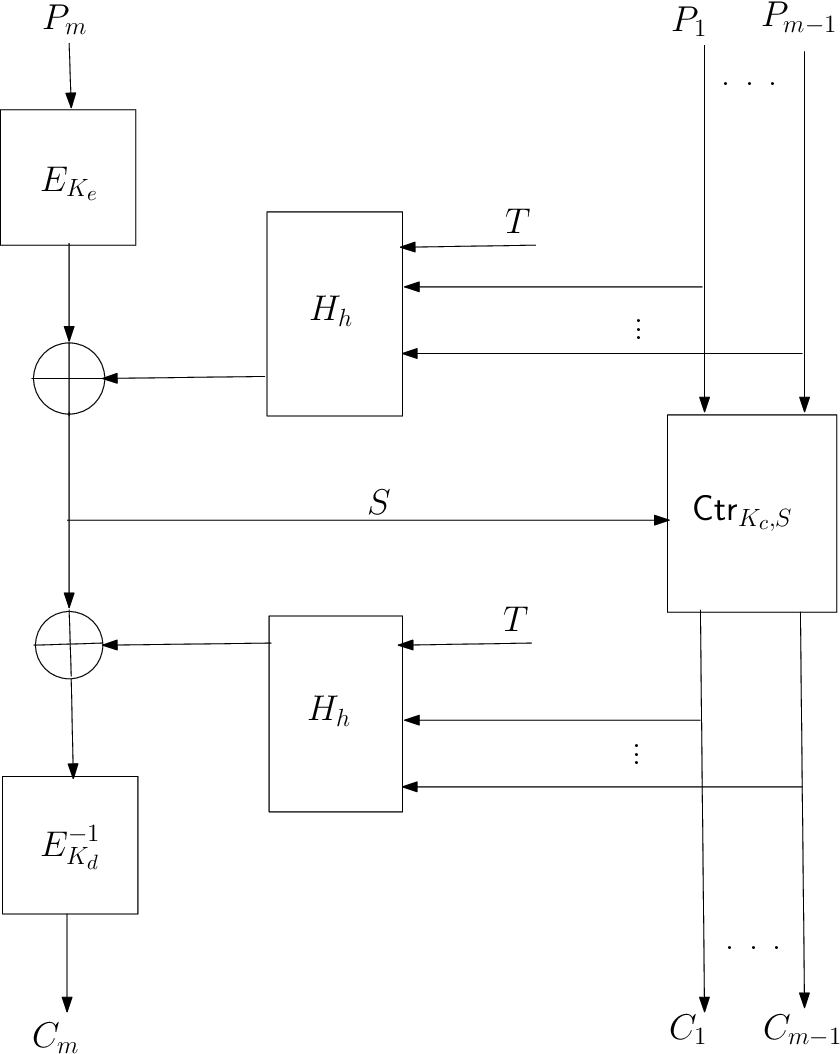}
	\centering
	\caption{\label{cons:xcbv2} Encryption of XCBv2}
\end{figure}

\begin{figure}[h]
	\begin{center}
		{\scriptsize
			\begin{tabular}{|c|}
				\hline\\
				\begin{minipage}{370pt}
					Encryption under XCBv1 : $\textbf{E}_K^T(P)$\\\\
					0.  $(P_1,\dotsc, P_m) \la \text{parse}_n P$\\
					1.  $h_1 \la E_K(0^{n-3}\|001)$\\
					2.  $h_2 \la E_K(0^{n-3}\|011)$\\
					3.  $K_e \la E_K(0^n)$\\
					4.  $K_d \la E_K(0^{n-3}\|100)$\\
					5.  $K_c \la E_K(0^{n-3}\|010)$\\
					6.  $CC \la E_{K_e}(P_1)$\\
					7.  $S \la CC \oplus H_{h_1}(P_2 \|\dotso \| P_{m-1}\|P_{m}, T)$\\
					8.  $(C_2, \dotsc,C_m) \la \text{{\sf Ctr}}_{K_{c,S}}(P_2,\dotsc, P_m) $\\
					9.  $MM \la S \oplus H_{h_2}(C_2 \|\dotso\| C_{m-1}\|C_m, T)$\\
					10. $C_1 \la E^{-1}_{K_d}(MM)$\\
					11. $\textbf{return}(C_1, C_2,\dotsc,C_m)$\\
				\end{minipage}\\
				\hline\\
				\begin{minipage}{370pt}
					Encryption under XCBv2 : $\textbf{E}_K^T(P)$\\\\
					100.  $P_m \la \text{lsb}_n(P)$\\
					101.  $A \la \text{msb}_{|P|-n}(P)$\\
					102.  $(P_1, P_2, \dotsc, P_{m-2}, P_{m-1}) \la \text{parse}_n(A)$\\
					103.  $\text{h} \la E_K(0^n)$\\
					104.  $K_e \la \text{msb}_{|K|}(E_{|K|}(0^{n-3}\|001)\|E_{|K|}(0^{n-3}\|010))$\\
					105.  $K_d \la \text{msb}_{|K|}(E_{|K|}(0^{n-3}\|011)\|E_{|K|}(0^{n-3}\|100))$\\
					106.  $K_c \la \text{msb}_{|K|}(E_{|K|}(0^{n-3}\|101)\|E_{|K|}(0^{n-3}\|110))$\\
					107. $CC \la E_{K_e}(P_m)$\\
					108. $S \la CC \oplus H_{h}(0^n\|T, P_1\|\dotso\|P_{m-2}\|\text{{\sf pad}}(P_{m-1})\|0^n)$\\
					109. $(C_1, \dotsc, C_{m-1} )\la \text{{\sf Ctr}}_{K_{c,S}}(P_1, \dotsc, P_{m-2}, P_{m-1})$\\
					110. $MM \la S\, \oplus \,H_h(T\|0^n, C_1 \| \dotso \| \text{{\sf pad}}(C_{m-1}) \| \text{\sf bin}_{\frac{n}{2}}(|T\|0^n|)\|\text{\sf bin}_{\frac{n}{2}}(|C_1\|\dotso \|C_{m-2}\|C_{m-1}|))$\\
					111. $C_m \la E^{-1}_{K_d}(MM)$\\
					112. $\textbf{return}(C_1, C_2,\dotsc, C_m)$\\
				\end{minipage}\\
				\hline
			\end{tabular}
		}
	\end{center}
	\caption{\label{fig : XCBv1_v2} Encryption using XCBv1 and XCBv2}
\end{figure}

\section{Differences between \text{XCBv1} and XCBv2}
\par The main difference between \text{XCBv1} with \text{XCBv2} is that later version uses only one hash key whereas the earlier one uses two hash keys which was made to reduce the cost of additional hash keys. There are some other differences which we can list as follows:  
\begin{enumerate}
	\item Keys generating by the same master key are different for XCBv1 and XCBv2. 
	\item In XCBv1, length of the key is fixed to 128 bits while XCBv2's key length is variable which can be 128, 192 or 256 bits as per requirement.
	\item As mentioned earlier, XCBv1 uses two hash keys as compare to XCBv2 which use only one hash key.
	\item In XCBv1, first block of the message is encrypted (in line 6) and then XOR with hash function (in line 7) while XCBv2 perform same operation with the last block of the plaintext.
	\item Padding is the part of XCBv2 hash function while XCBv1's hash function has no padding. 
	\item XCBv2 append the string of $0^n$ before tweak and after message in hash function (in line 108) while XCBv1 is simple hash function without appending any string with tweak and message.
	\item In XCBv1, both hash function use the different hash key and same definition of hash function while XCBv2 uses same hash key and different hash function in the scheme (in line 108 and 110).
\end{enumerate}
All keys are supposed to be random by the security of AES.



\section{Security Claims}
In this section, we state both security bounds as claimed in \cite{chakraborty} and our updated security bounds for both XCBv1 and XCBv2. We are interested in so called information theoretic security bound, i.e., we only state the bound for the schemes where the block ciphers are replaced by true random permutations. In both versions of XCB three block ciphers with different keys are used. We replace these block ciphers with three independent random permutations on $\nbits$, we call the resulting construction as $\text{XCBvl}[3\text{Perm}(n)]$ where $\text{l} \in \{1, 2\}$. For XCBv1 security bound as stated in \cite{chakraborty} is:

\begin{theorem}
	Let $A$ be an arbitrary adversary which queries only with messages/ciphers whose lengths are multiples of $n$ and $A$ asks a total of $q$ queries of overall query complexity $\sigma_n$ where each query is at most $\ell$ blocks long (each block of $n$ bits). Then,
	$${\bf Adv}^{\pm\nprp}_{\text{XCBv1}[3\text{Perm}(n)]}(A) \leq \frac{(3+2^{22})\ell q \sigma_n}{2^n}.$$
\end{theorem}

For XCBv2 the security bound given in \cite{chakraborty} is:
\begin{theorem}
	Let $A$ be an arbitary adversary which queries only with messages/ciphers whose lengths are multiples of $n$ and $A$ asks a total of $q$ queries of overall query complexity $\sigma_n$ where each query is at most $\ell$ blocks long (each block of $n$ bits). Then,
	$${\bf Adv}^{\pm\nprp}_{\text{XCBv2fb}[3\text{Perm}(n)]}(A) \leq \frac{(5+2^{22})\ell q \sigma_n}{2^n}.$$
\end{theorem}

\par We show that the bound of $\frac{(3+2^{22})\ell q \sigma_n}{2^n}$ and $\frac{(5+2^{22})\ell q \sigma_n}{2^n}$ in both the above theorems can be improved to $\frac{(3+2^5)\ell q \sigma_n}{2^n}$ and $\frac{(5+2^5)\ell q \sigma_n}{2^n}$ for XCBv1 and XCBv2 respectively. Specifically, we prove the following two theorems.\\
For XCBv1, our security bound is as follows:
\begin{theorem}\label{imp_thm_xcbv1}
	Consider an arbitrary adversary $A$ which queries only with messages/ciphers whose lengths are multiples of $n$ and $A$ asks a total of $q$ queries of overall query complexity $\sigma_n$ where each query is at most $\ell$ blocks long (each block of $n$ bits). Then,
	\begin{equation}\label{imp_xcbv1}
	{\bf Adv}^{\pm\nprp}_{\text{XCBv1}[3\text{Perm}(n)]}(A) \leq \frac{(3+2^5)\ell q \sigma_n}{2^n}.
	\end{equation}
\end{theorem}

and for XCBv2 our security bound is given as:

\begin{theorem}\label{imp_thm_xcbv2}
	Consider an arbitrary adversary $A$ which queries only with messages/ciphers whose lengths are multiples of $n$ and $A$ asks a total of $q$ queries of overall query complexity $\sigma_n$ where each query is at most $\ell$ blocks long (each block of $n$ bits). Then,
	\begin{equation}\label{imp_xcbv2}
	{\bf Adv}^{\pm\nprp}_{\text{XCBv2fb}[3\text{Perm}(n)]}(A) \leq \frac{(5+2^5)\ell q \sigma_n}{2^n}.
	\end{equation}
\end{theorem}

The complete proof of these theorems are presented in the following section. The proof heavily uses a technique from \cite{Iwata} where an improved bound on GCM was proved.

\section{Some useful lemmas}
We start with a few useful lemmas which would be used later.\\ 
The following Version of the Schwartz-Zippel lemma is from \cite{rajeev}.

\begin{lemma}\label{szl}
	Let $\mathbb{F}$ be a field and $p \in \mathbb{F}[x_1, x_2,\dotsc, x_r]$ be a $r$-variate, non-zero polynomial with total degree $d$. Let $S$ be finite subset of $\mathbb{F}$, and $x_1, x_2,\dotsc,x_r$ be selected uniformly at random from $S$. Then 
	$$ \Pr[p(x_1, x_2,\dotsc, x_r) = 0] \leq \dfrac{d}{|S|}.$$
\end{lemma}

For $0 \leq r \leq 2^{32}-1$,
$\mathbb{Y}_r \api \{{\sf bin}_{32} (\text{int}(Y) + r \mod 2^{32}) \oplus Y | Y \in \tbits \}$ and $\text{inc}^r(X) = (X \oplus 0^{n-32}\|Y)$ for some $Y \in \mathbb{Y}_r$. \\
From \cite{Iwata} we define $\mathbb{W}_r \subseteq \tbits$, for $ 0 \leq r \leq 2^{32}-1$, as
$\mathbb{W}_0 \api \mathbb{Y}_0$ and $\mathbb{W}_r \api \mathbb{Y}_r \backslash \bigcup_{i=0}^{r-1} \mathbb{Y}_i $ and $r \ge 1$. We denote cardinality as $w_r \api \# \mathbb{W}_r$ and $w_{\max} \api max\{w_r \ | \ 0 \leq r \leq 2^{32}-1\}$ and it was shown in \cite{Iwata} that $w_{\max} \leq 32$.\\
\begin{lemma}\label{lem:2}
	\textbf{1.} Let $X, Y, X', Y' \in  \mbits$, such that $(X,Y) \neq (X', Y')$. Let $C, C' \in \nbits$ and $h \rand \nbits, S = C \oplus H_h(X, Y),$ and $S' = C' \oplus H_h(X', Y'),$ where $H_h(\cdot)$ is defined in (\ref{1}). Then,
	
	$$\Pr\bigg[\bigvee_{i=0}^{m^s-2}\bigg(\text{inc}^i(S)\oplus S' = 0\bigg)\bigg] \leq \dfrac{w_{max}{\ell}(m^s-1)}{2^n}.$$
	
	\textbf{2.} Let $X, Y, X', Y' \in  \mbits$, $C, C' \in \nbits$ and $h_1, h_2 \rand \nbits, S = C \oplus H_{h_1}(X, Y),$ and $S' = C' \oplus H_{h_2}(X', Y')$. Then,
	
	$$\Pr\bigg[\bigvee_{i=0}^{m^s-2}\bigg(\text{inc}^i(S)\oplus S' = 0\bigg)\bigg] \leq \dfrac{w_{max}{\ell}(m^s-1)}{2^n}.$$
	
	In both cases $\ell$ is the $\text{max}\{\ell^s,\ell^{s'}\}$ where $\ell^s$ and $\ell^{s'}$ are the degrees of the two polynomials $S$ and $S'$ respectively. In the first case probability is taken over the random choice of $h$, and in the second case it is taken over the random choice of $h_1, h_2$.
\end{lemma}

\begin{proof}
	\textbf{Case 1.} The proof uses technique given in \cite{Iwata}.\\
	
	Let's figure out the upper bound on collision with two distinct pair $(X,Y)$ and $(X', Y')$ with $0 \leq r < r' \leq 2^{32} -1 $ i.e. 
	$$ \Pr[(\text{inc}^r(S) = S') \vee (\text{inc}^{r'}(S) = S')]$$
	
	Also we obtain the following upper sum bound
	\begin{multline}
	\Pr[(\text{inc}^r(S) = S') \vee (\text{inc}^{r'}(S) = S')] \leq \sum_{Y \in \mathbb{Y}_r} \Pr[S \oplus (0^{n-32}\|Y) = S']  \\ + \sum_{Y' \in \mathbb{Y}_{r'}} \Pr[S \oplus (0^{n-32}\|Y') = S'] 
	\end{multline}
	since $\text{inc}^r(X) = (X \oplus 0^{n-32}\|Y)$ for some $Y \in \mathbb{Y}_r.$\\
	
	\textbf{Claim:} For $ 0 \leq r'  <   r \leq 2^{32}-1$, and $Y \in \tbits$ such that $Y \in \mathbb{Y}_r$ and $Y \in \mathbb{Y}_{r'}$. Then there does not exist $X \in \nbits$ such that $\text{inc}^r(X) = X \oplus (0^{n-32}\|Y)$ and $\text{inc}^{r'}(X) = X \oplus (0^{n-32}\|Y)$ simultaneously.\\
	
	\textbf{Proof:} Proof by contradiction. Let's such $r$ and $r'$ exist, without loss of generality $r' < r$. Therefore, $\text{inc}^r(X) = \text{inc}^{r'}(X)$ which imply $\text{inc}^{r-r'}(X) = X$ which is not possible. Hence, $r$ and $r'$ are not distinct.\\
	
	So, we can conclude with an upper bound i.e.
	\begin{multline}
	\Pr[(\text{inc}^r(S) = S') \vee (\text{inc}^{r'}(S) = S')] \leq \sum_{Y \in \mathbb{Y}_r} \Pr[S \oplus (0^{n-32}\|Y) = S'] \\ + \sum_{Y' \in \mathbb{Y}_{r'} \backslash { \mathbb{Y}_r} } \Pr[S \oplus (0^{n-32}\|Y') = S']. 
	\end{multline}
	
	After generalisation for $m^s-2$, we get
	
	\begin{equation}
	\Pr\bigg[\bigvee_{i=0}^{m^s-2}\bigg(\text{inc}^i(S)\oplus S' = 0\bigg)\bigg] \leq \sum_{0 \leq i \leq m^s-2} \sum_{Y \in \mathbb{Y}_i \backslash \bigcup_{j=0}^{i-1} \mathbb{Y}_j } \Pr[S \oplus (0^{n-32}\|Y) = S']. \label{first} 
	\end{equation}
	Now, by using Schwartz-Zippel Lemma, we know
	$$\Pr[S \oplus (0^{n-32}\|Y) = S'] \leq \dfrac{\ell}{2^n}. $$
	Therefore,
	\begin{equation}
	\sum_{0 \leq i \leq m^s-2} \sum_{Y \in \mathbb{Y}_i \backslash \bigcup_{j=0}^{i-1} \mathbb{Y}_j } \Pr[S \oplus (0^{n-32}\|Y) = S'] \leq \sum_{0 \leq i \leq m^s-2} \dfrac{w_{max}{\ell}}{2^n}.\label{second}
	\end{equation}
	From equations (\ref{first}) and (\ref{second}), we can conclude
	$$\Pr\bigg[\bigvee_{i=0}^{m^s-2}\bigg(\text{inc}^i(S)\oplus S' = 0\bigg)\bigg] \leq \dfrac{w_{max}{\ell}(m^s-1)}{2^n}.$$

\end{proof}

\textbf{Case 2.} $h_1, h_2$ are selected independently and uniformly at random from $\nbits$, and $S = H_{h_1}(X,Y)$ and $S' = H_{h_2}(X',Y')$. According to the definition of $H(\cdot)$, both $S$ and $S'$ are non-zero polynomials. So, the result will be same as of Case 1 by Lemma \autoref{szl}.


\section{Repairing XCB Security proofs}	
Proof of the theorem is heavily based on the proof given in \cite{chakraborty}. As stated earlier, in place of three block cipher given in line 6,8 and 10 of XCBv1, we use the three different permutation on $n$-bit string. The encryption and decryption of the scheme of $\text{XCBv1}[3\text{Perm}(n)]$ by $\textbf{E}_{\tilde{\pi},\tilde{h}}$ and $\textbf{D}_{\tilde{\pi},\tilde{h}}$ respectively, where $\tilde{\pi} = (\pi_1, \pi_2, \pi_3)$  and $\pi_1, \pi_2, \pi_3$ are three permutation selected uniformly and independently  and $\tilde{h} = (h_1, h_2)$ where $h_1$ and $h_2$ are two hash keys selected uniformly and independently to other variables.

\par For proving (\ref{imp_xcbv1}), we need to consider an adversary's advantage in distinguishing $\text{XCBv1}[3\text{Perm}(n)]$ from an oracle which simply returns random bit strings. This advantage defined in following way :
\begin{equation*}
\textbf{Adv}^{\pm rnd}_{\text{XCBv1}[3\text{Perm}(n)]}(A) = \bigg|\Pr\bigg[\tilde{\pi} \rand 3\text{Perm}(n), \tilde{h} \rand \nbits : A^{\textbf{E}_{\tilde{\pi},\tilde{h}},\textbf{D}_{\tilde{\pi},\tilde{h}}}\Rightarrow 1\bigg]
\end{equation*}
\begin{equation}
\quad - \Pr\bigg[A^{\$(\cdot,\cdot),\$(\cdot,\cdot)}\Rightarrow 1\bigg]\bigg|,
\end{equation}
where $\$(\cdot,M)$ or $\$(\cdot,C)$ returns independently distributed random bits of length $|M|$ or $|C|$  respectively. The basic idea of proving (\ref{imp_xcbv1}) is as follows.

\begin{eqnarray}
\nonumber
\textbf{Adv}^{\pm \widetilde{prp}}_{\text{XCBv1[3Perm}(n)]}(A) &=& \bigg(\Pr\bigg[\tilde{\pi} \rand \text{3Perm}(n), \tilde{h} \rand \nbits : A^{\textbf{E}_{\tilde{\pi},\tilde{h}},\textbf{D}_{\tilde{\pi},\tilde{h}}}\Rightarrow 1\bigg]\\\nonumber
&&- \Pr\bigg[\pi \rand \text{Perm}^{\mathcal{T}}{\mathcal{(M)}}:A^{\pi(\cdot,\cdot),\pi^{-1}{(\cdot,\cdot)}}
\Rightarrow 1\bigg]\bigg)\\\nonumber
&=& \bigg(\Pr\bigg[\tilde{\pi} \rand \text{3Perm}(n), \tilde{h} \rand \nbits : A^{\textbf{E}_{\tilde{\pi},\tilde{h}},\textbf{D}_{\tilde{\pi},\tilde{h}}}\Rightarrow 1\bigg]\\\nonumber
&& - \Pr\bigg[A^{\$(\cdot, \cdot),\$(\cdot, \cdot)}\Rightarrow 1\bigg]\bigg)\\\nonumber
&& + \bigg(\Pr\bigg[A^{\$(\cdot, \cdot),\$(\cdot, \cdot)}\Rightarrow 1\bigg]\\\nonumber
&& - \Pr\bigg[\pi \rand Perm^{\mathcal{T}}{\mathcal{(M)}}:A^{\pi(\cdot,\cdot),\pi^{-1}{(\cdot,\cdot)}}
\Rightarrow 1\bigg]\bigg)\\
&\leq& \textbf{Adv}^{\pm rnd}_{\text{XCBv1[3Perm}(n)]}(A) + \binom{q}{2}\frac{1}{2^n}\label{idea}.
\end{eqnarray}
where $q$ is the number of queries made by the adversary. For proof of the last inequality see \cite{halevi2004parallelizable}. Thus, the main task of the proof now reduces to obtaining an upper bound on $\textbf{Adv}^{\pm rnd}_{\text{XCBv1[3Perm}(n)]}(A)$.
We prove this by the usual techniques of sequence of games which are in games XCB1 (\autoref{fig : XCB1_RAND1}), RAND1 (Figure \ref{fig : XCB1_RAND1}) and RAND2 (\autoref{fig : XCBv1_RAND2}).

\textit{Game XCB1} 
is same as in \autoref{fig : XCB1_RAND1} except the algorithm of XCBv1 uses three independent random permutation $\pi_1, \pi_2, \pi_3$ instead of the block cipher implementation. We denote this as follows:
\begin{equation}
\Pr\bigg[A^{\textbf{E}_{\tilde{\pi},h},\textbf{D}_{\tilde{\pi},h}}\Rightarrow 1\bigg] = \Pr\bigg[A^\text{XCB1} \Rightarrow 1\bigg].
\end{equation}

\begin{figure}
	\begin{small}
		\begin{center}
			{\begin{adjustbox}{width=1\textwidth}
					\begin{tabular}{|c|}
						\hline\\
						\begin{minipage}{430pt}
							Subroutine Ch-$\pi_i(X) (i = 1, 2, 3)$\\
							11. $ Y \rand \nbits$; \textbf{if}\;$ Y \in Range_i$ \textbf{then} {\sf bad}  $\la$ \textbf{true};
							\fbox{$Y \rand \overline{Range_i}$}; \textbf{end if} \\
							12. \textbf{if} $X \in Domain_i$ \textbf{then \text{{\sf bad}}} $\la$ true; \fbox{$Y \la \pi_i(X)$}; \textbf{end if} \\
							13.  $\pi_i(X) \la Y; Domain_i \la Domain_i \cup \{X\}$; $ Range_i \la Range_i  \cup (Y)$; \textbf{return} (Y);\\\\
							Subroutine Ch-$\pi_i^{-1}(Y) $\\  
							14.  $ X \rand \nbits$; \textbf{if}\;$ X \in Domain_i$ \textbf{then} \textbf{{\sf bad}} $\la$ \textbf{true};
							\fbox{$X \rand \overline{Domain_i}$}; \textbf{end if} \\
							15. \textbf{if} $Y \in Range_i$ \textbf{then \text{{\sf bad}}} $\la$ true; \fbox{$X \la \pi_i^{-1}(Y)$}; \textbf{end if} \\
							16.  $\pi_i(X) \la Y; Domain_i \la Domain_i \cup \{X\}$; $ Range_i \la Range_i  \cup (Y)$; \textbf{return} (X);\\
							
							\underline{Initialization:}\\
							17.  \textbf{for all} $X \in \nbits \pi_i(X)$ = undef \textbf{end for}\\
							18. {\sf bad} = false\\
							19. $h_1, h_2 \rand \nbits$\\
						\end{minipage}\\
						\hline
						\begin{minipage}{430pt}
							Respond to the $s^{th}$ query as follows:
						\end{minipage}\\
						\hline
						\begin{minipage}{430pt}
							\begin{tabular}{l|l}
								\underline{Encipher query:} $Enc(T^s; P_1^s, P_2^s,\dotsc,P_{m_s}^s)$ &\underline{Decipher query:} $Dec(T^s; C_1^s, C_2^s, \dotsc, C_{m^s}^s)$\\\\
								
								101. \textbf{if}  $P_{1}^s = P_{1}^{s'}$ \textbf{for} $s' < s$ \textbf{then} & 101. \textbf{if} $C_{1}^s = C_{1}^{s'}$ \textbf{for} s' < s \textbf{then}\\
								102. $CC^s \la CC^{s'}$& 102. $MM^s \la MM^{s'}$\\
								103. \textbf{else} & 103. \textbf{else}\\
								104. $CC^s \la \text{Ch-}\pi_1(P_{1}^s)$ & 104. $MM^s \la \text{Ch-}\pi_3(C^s_{m^s})$\\
								105. \textbf{end if} & 105. \textbf{end if}\\
								106. $S^s \la CC^s \oplus H_{h_1}(P_2^s \| \dotso \| P_{m^s}^s,T^s)$ & 106. $S^s \la MM^s \oplus H_{h_2}(C_2^s \| \dots \|C_{m^s}^s, T^s)$\\
								107. \textbf{for} $i = 1$ to $m^s$ -1 & 107. \textbf{for} $i = 1$ to $m^s$-1\\
								
								108. $Z_i^s \la \text{Ch-}\pi_2(\text{inc}^i(S^s))$ & 108. $C_{i+1}^s \la P_{i+1}^s \oplus Z_i^s$\\
								109. $C_{i+1}^s \la P_{i+1}^s \oplus Z_i^s$ & 109. $P_{i+1}^s \la C_{i+1}^s \oplus Z_i^s$\\
								110. \textbf{end for} & 110. \textbf{end for}\\
								111. $MM^s \la S^s \oplus H_{h_2}(C_2^s \| \dots \|C_{m^s}^s, T^s)$ & 111. $CC^s \la S^s \oplus  H_{h_1}(P_2^s \| \dotso \| P_{m^s}^s,T^s)$\\
								112. $C_1^s \la Ch-\pi_3^{-1}(MM^s)$ & 112. $P_{1}^s \la Ch-\pi_1^{-1}(CC^s)$\\
								113. \textbf{return} $(C_1^s, C_2^s, \dotsc, C_{m^s}^s)$ & 113. \textbf{return} $(P_1^s, P_2^s, \dotsc, P_{m^s}^s)$\\ 
							\end{tabular}
						\end{minipage}\\
						\hline
					\end{tabular}
				\end{adjustbox}
			}
			\caption{\label{fig : XCB1_RAND1} Games XCB1 and RAND1 : In RAND1 the boxed entries are removed.}
		\end{center}
	\end{small}
	
\end{figure}

\begin{figure}
	
	\begin{center}
		{\footnotesize
			\begin{adjustbox}{width=1\textwidth}
				\begin{tabular}{|c|}
					\hline
					\begin{minipage}{430pt}
						\begin{tabular}{l}
							Respond to the $s^{th}$ adversary query as follows:
							\\
							ENCIPHER QUERY $\textbf{Enc}(T^s; P^s)$\\
							10. $(P_1^s, P_2^s, \dotsc, P_{m^s}^s) \la \text{parse}_n(P^s)$\\
							11. $ty^s = \textbf{Enc}$\\
							12. $C_1^s\|C_2^s\|\dotso\|C_{m^s-1}^s\|D_{m^s}^s \rand
							\{0,1\}^{nm^s}$\\
							13. $C_{m^s}^s \la \text{\sf drop}_{n-r^s}(D_{m^s}^s)$\\
							14. \textbf{return} $C_1^s\|C_2^s\|\dotsc\|C_{m^s}^s$\\
							DECIPHER QUERY $\textbf{Dec}(T^s; C^s)$\\
							20. $(C_1^s, C_2^s,\dotsc, C_{m^s-1}^s, C_{m^s}^s) \la \text{parse}_n(C^s)$\\
							21. $ty^s = \textbf{Dec}$\\
							22.	$P_1^s\|P_2^s\|\dotsc\|P_{m^s-1}^s\|V_{m^s}^s \rand \{0,1\}^{nm^s}$\\
							23. $P_{m^s}^s \la \text{\sf drop}_{{n-r}^s}(V_{m^s})$\\
							24. \textbf{return} $P_1^s\|P_2^s\|\dotso\|P_{m^s}^s$
						\end{tabular}
					\end{minipage}\\
					\hline
					\begin{minipage}{430pt}
						\textbf{Finalization:}\\
						001. $h_1 \rand \nbits$;\qquad
						002. $h_2 \rand \nbits$
					\end{minipage}\\
					\hline
					\begin{minipage}{430pt}		
						\begin{tabular}{l|l}			
							\textbf{for} s = 1 to q\\
							\qquad \textbf{if}\;$ty^s$ = \textbf{Enc} \textbf{then} & \quad \textbf{else if} $ty^s$ = \textbf{Dec}:\\
							101. \textbf{if} $P_{1}^s = P_{1}^{s'}$ \textbf{for} s' < s \textbf{then} & 201. \textbf{if} $C_1^s = C_{1}^{s'}$ \textbf{for} s' < s then\\
							
							102. \quad$CC^s \la CC^{s'}$ & 202. $MM^s \la MM^{s'}$ \\
							
							103. \textbf{else} & 203. \textbf{else}\\
							
							104. \quad $CC^s \rand \nbits$ & 204. \quad$MM^s \rand \nbits$\\
							
							105. \quad$\mathcal{D}_1 \la \mathcal{D}_1 \cup \{P_{1}^s\}$ & 205. \quad$\mathcal{D}_3 \la \mathcal{D}_3 \cup \{C_{1}^s\}$\\
							
							106. \quad$\mathcal{R}_1 \la  \mathcal{R}_1 \cup \{CC^s\}$ & 206. \quad$\mathcal{R}_3 \la \mathcal{R}_3 \cup \{MM^s\}$\\
							107. \textbf{end if} & 207. \textbf{end if}\\
							
							108. $S^s \la CC^s \oplus H_{h_1}(P_2^s \| \dotso \| P_{m^s}^s,T^s)$ & 208. $S^s \la MM^s \oplus H_{h_2}(C_2^s \| \dotso \| C_{m^s}^s,T^s)$\\
							
							109. $MM^s \la S^s \oplus H_{h_2}(C_2^s \| \dotso \| C_{m^s}^s,T^s)$ & 209. $CC^s \la S^s \oplus H_{h_1}(P_2^s \| \dotso \| P_{m^s}^s,T^s)$\\
							
							110. $\mathcal{D}_3 \la \mathcal{D}_3 \cup \{C_{1}^s\}$ & 210. $\mathcal{D}_1 \la \mathcal{D}_1 \cup \{P_{1}^s\}$\\
							
							111. $\mathcal{R}_3 \la \mathcal{R}_3 \cup \{MM^s\}$ & 211. $\mathcal{R}_1 \la \mathcal{R}_1 \cup \{CC^s\}$\\
							
							112. \textbf{for} $i = 0$ to $m^s-3$, & 212. \textbf{for} $i = 0$ to $m^s-3$,\\
							
							113. \quad$Y_i^s \la C_{i+2}^s \oplus P_{i+2}^s$ & 213. \quad$Y_i^s \la C_{i+2}^s \oplus P_{i+2}^s$\\
							
							114. \quad$\mathcal{D}_2 \la \mathcal{D}_2 \cup \{\text{inc}^i(S^s)\}$ & 214. \quad$\mathcal{D}_2 \la \mathcal{D}_2 \cup \{\text{inc}^i(S^s)\}$\\
							
							115. \quad$\mathcal{R}_2 \la \mathcal{R}_2 \cup \{Y_i^s\}$ & 215. \quad$\mathcal{R}_2 \la  \mathcal{R}_2 \cup \{Y_i^s\}$\\
							
							116. \textbf{end for} & 216. \textbf{end for}\\
							117.  $	Y_{m^s-2}^s \la \text{{\sf pad}}(P_{m^s}^s) \oplus D_{m^s}^s$ & 217.  $	Y_{m^s-2}^s \la \text{{\sf pad}}(C_{m^s}^s) \oplus V_{m^s}^s$\\	
							
							118. $\mathcal{D}_2 \la \mathcal{D}_2 \cup \{\text{inc}^{m^s-2}(S^s)\}$ & 218. $\mathcal{D}_2 \la \mathcal{D}_2 \cup \{\text{inc}^{m^s-2}(S^s)\}$\\
							
							119. ${\sf \mathcal{R}}_2 \la {\sf \mathcal{R}}_2 \cup \{Y_{m^s-2}^s\}$ & 219. ${\sf \mathcal{R}}_2 \la {\sf \mathcal{R}}_2 \cup \{Y_{m^s-2}^s\}$\\
							&\qquad \textbf{end if}\\
							&\textbf{end for}
						\end{tabular}
					\end{minipage}\\
					\hline
					\begin{minipage}{430pt}
						SECOND PHASE\\
						\quad\text{{\sf bad} = false};\\
						\quad \textbf{if}(some value occurs more than once in 
						$\mathcal{D}_{i}, i = 1, 2, 3) \textbf{then {\sf bad}} = true$ \textbf{end if;}\\
						\quad \textbf{if}(some value occurs more than once in 
						$\mathcal{R}_{i}, i = 1, 2, 3) \textbf{then {\sf bad}} = true$ \textbf{end if.}\\
					\end{minipage}\\
					\hline
				\end{tabular}
			\end{adjustbox}
		}
	\end{center}
	\caption{\label{fig : XCBv1_RAND2}Game RAND2 for XCBv1}
\end{figure}
\textit{Game RAND1} is also described in \autoref{fig : XCB1_RAND1} with the boxed entries removed. In this game it is not guaranteed that $\pi_i(i = 1, 2, 3)$ are permutation as though we do the consistency checks but we don't reset the values of $Y$ (in Ch-$\pi_i$) and $X$ (in Ch-$\pi_i^{-1}$). Thus, the games XCB1 and RAND1 are identical apart from what happens when the \text{{\sf bad}} flag is set. By the fundamental lemma of game-ploting or difference lemma, we have
\begin{equation}
\bigg|\Pr\bigg[A^\text{XCB1} \Rightarrow 1\bigg] - \Pr\bigg[A^{RAND1} \Rightarrow 1\bigg] \bigg| \leq \Pr\bigg[A^{RAND1} \; \text{set} \; \text{{\sf bad}}\bigg].
\end{equation}
Here, we see RAND1 gives the random string in response of encryption and decryption queries. So,
\begin{equation}
\Pr\bigg[A^{RAND1} \Rightarrow 1\bigg] = \Pr\bigg[A^{(\cdot,\cdot),(\cdot,\cdot)} \Rightarrow 1\bigg].
\end{equation}
Now, by using the definition 
\begin{eqnarray}
\nonumber
\textbf{Adv}^{\pm rnd}_{\text{XCBv1[3Perm}(n)]}(A) &=& 
\bigg| \Pr\bigg[A^{\textbf{E}_{\pi_1, \pi_2, \pi_3}, \textbf{D}_{{\pi_1, \pi_2, \pi_3}}} \Rightarrow 1\bigg]  \\\nonumber
&& - \Pr\bigg[A^{(\cdot,\cdot),(\cdot,\cdot)} \Rightarrow 1\bigg]\bigg|\\\nonumber
&=& \bigg|\Pr\bigg[A^\text{XCB1} \Rightarrow 1\bigg]\\\nonumber
&& - \Pr[A^{RAND1} \Rightarrow 1\bigg] \bigg|\\
&\leq&\label{15} \Pr\bigg[A^{RAND1} \; \text{set} \; {\sf bad}\bigg]. 
\end{eqnarray}

\textit{Game RAND2} is slightly different from RAND1, in this permutation is not maintained, just a random string of appropriate length in response of an encryption/decryption query is returned. In the finalisation step  of game, the internal variable are adjusted and the appropriate variables are inserted in the multi sets $ \mathcal{D}_1, \mathcal{D}_2, \mathcal{D}_3$ and $\mathcal{R}_1, \mathcal{R}_2, \mathcal{R}_3$. If collision occurs in these multi sets then the {\sf bad} flag is set.

Games RAND1 and RAND2 are indistinguishable to the adversary as both returns the random strings in response to queries. And also for both the cases, probability for which RAND1 and RAND2 have {\sf bad} flag set is same. Therefore, we can write:

\begin{equation}\label{16}
\Pr[A^{RAND1} \; \text{set}\; {\sf bad}] = \Pr[A^{RAND2} \; \text{set}\; {\sf bad}].  
\end{equation}
Thus, from equations (\ref{15}) and (\ref{16})
\begin{equation}\label{idea complete}
\textbf{Adv}^{\pm rnd}_{\text{XCBv1[3Perm}(n)]}(A) \leq \Pr[A^{RAND2} \; \text{set} \; {\sf bad}]. 
\end{equation}

So, our goal is to bound $\Pr[A^{RAND2} \; \text{set} \; {\sf bad}]$. If there is a collision in these multi sets in Game RAND2 then the {\sf bad} flag is set. So if $\ {\sf COLLD}_i$ and ${\sf COLLR}_i$ denote the events of a collision in $\mathcal{D}_i$ and $\mathcal{R}_i$ respectively then we have
\begin{equation}\label{18}
\Pr[A^{RAND2} \; \text{set} \; {\sf bad}] \leq \sum_{1 \leq i \leq 3}(\Pr[{\sf COLLR}_i] + \Pr[{\sf COLLD}_i]).
\end{equation}
In the rest of the section we analyze the collision probabilities in the sets $\mathcal{D}_i$ and $\mathcal{R}_i$. After $q$ queries of the adversary where the $s^{th}$ query has $m^s$ blocks of plaintext or ciphertext and $t^s$ blocks of tweak, then the sets $\mathcal{D}_i$ and $\mathcal{R}_i$ can be written as follows:\\\\
$\mathcal{D}_1 = \{P_1^s : 1 \leq s \leq q\},$\\
$\mathcal{D}_2 = \bigcup_{s=1}^q\{inc^j(S^s):0 \leq j \leq m^s-2\},$\\
$\mathcal{D}_3 =\{C_1^s : 1 \leq s \leq q\},$\\\\
$\mathcal{R}_1 =\{CC^s : 1 \leq s \leq q\},$\\
$\mathcal{R}_2 = \bigcup_{s=1}^q\{Y_j^s = C_{j+2}^s \oplus P_{j+2}^s :0 \leq j \leq m^s-3\},$\\
$\mathcal{R}_3 = \{MM^s : 1 \leq s \leq q\}.$\\\\

Following are the points which will help in the analysis:
\begin{enumerate}
	\item For the $s^{th}$ query $ty^s \in \{enc, dec\}$ will denote whether the query is an encryption or a decryption query.
	\item In each query, the adversary specifies a tweak $T^s$, we consider $t^s = \lceil(\left|T^s\right|/n)\rceil$. Thus, for any s, $H_{h_1}$ and $H_{h_2}$ in line 108 and  109 respectively ($H_{h_1}$ and $H_{h_2}$ in line 209 and 208 respectively for decryption) for encryption of game RAND2 has degree at most $m^s + t^s$. We denote $\sigma_n = \sum_{s}t^s + \sum_{s}m^s$. We denote $\ell = \text{max}\{m^s + t^s,m^{s'} + t^{s'}\}$.
	\item In game RAND2 the hash key $h_1$ and $h_2$ are selected uniformly at random from $\nbits$.
	\item For an encryption query, the response received by $A$ is $(C_1^s, C_2^s,\dotsc,C_{m^s}^s)$ and for a decryption query the response received is $(P_1^s, P_2^s,\dotsc,P_{m^s}^s)$. Both these responses are uniformly distributed and independent of other variables.
\end{enumerate}
In the following claims we bound the required collision probabilities.
\begin{claim}\label{claim:5}
	$\Pr[{\sf COLLD}_1] \leq \binom{q}{2}/{2^n}.$
	\begin{proof}	
		In case of encryption i.e. $ty^s = ty^{s'} =$\textit{enc} then $\Pr[P_{1}^s = P_{1}^{s'}] = 0$ because of the condition in line 101 of RAND2.
		In case of at least one decryption. Without loss of generality, if $ty^s$ = \textit{dec}, then $P_1^s$ is a uniform $n$-bit string, hence $\Pr[P_1^s = P_1^{s'}] = 1/2^n$.\\
		For $q$ queries, $\Pr[{\sf COLLD}_1] \leq \binom{q}{2}/2^n.$
	\end{proof}
\end{claim}

\begin{claim}\label{claim:6}
	$\Pr[{\sf COLLD}_2] \leq (\ell q \sigma_n)w_{max}/2^n.$
	\begin{proof}
		$\mathcal{D}_2 = Z_1 \cup Z_2 \cup \dotso \cup Z_q $, where\\
		\qquad$Z_s = \{\text{inc}^j(S^s) : 0 \leq j \leq m^s-2\},$ for $1\leq s \leq q$
		and 
		\[S^s =
		\begin{cases}
		CC^s \oplus H_{h_1}(P_2^s \| \dotso \| P_{m^s}^s, T^s) & \text{if}\; ty^s = enc,\\
		MM^s \oplus H_{h_2}(C_2^s \| \dots \|C_{m^s}^s, T^s) & \text{if}\; ty^s = dec.
		\end{cases}
		\]
		
		As $Z_s$ contains distinct elements, if $ \forall x_1, x_2 \in Z_s$ then $\Pr[x_1 =  x_2] = 0$.
		We need to bound collision between $x1$ and $x2$ when these are in different set then without loss of generality $x_1 \in Z_s$  and $x_2 \in Z_{s'}$, for $s \neq s'$.
		For $s \neq s'$, we define ${\sf COLL}(Z_s, Z_{s'})$  as the event that at least one element of $Z_s$ collide with one element of $Z_{s'}$. Hence, we have
		\begin{equation}\label{27}
		\Pr[{\sf COLLD}_2] \leq \sum_{1 \leq s < s' \leq q} \Pr[{\sf COLL}(Z_s, Z_{s'})].
		\end{equation}
		We also have\\
		\[\text{inc}^j(S^s) \oplus \text{inc}^{j'}(S^{s'}) =
		\begin{cases}
		\text{inc}^{j-j'}(S^s) \oplus S^{s'} & \text{if}\; j \geq j'.\\
		S^{s}  \oplus \text{inc}^{j'-j}(S^{s'}) & \text{if}\; j < j'.
		\end{cases}
		\]
		Now we define the following events
		\begin{equation}
		U_i \equiv \text{inc}^{i}(S^s) \oplus S^{s'} = 0, \quad \text{for}\; 0 \leq i \leq m^s-2,
		\end{equation}
		\begin{equation}
		V_i \equiv S^{s}  \oplus \text{inc}^{i}(S^{s'}) = 0 \quad \text{for}\; 0 \leq i \leq m^{s'}-2.
		\end{equation}
		Hence,
		\begin{equation}\label{30}
		{\sf COLL}(Z_s, Z_{s'}) = \Bigg(\bigvee_{i=0}^{m^s-2}U_i\Bigg)\bigvee \Bigg(\bigvee_{i=0}^{m^{s'}-2}V_i\Bigg).
		\end{equation}
		Now we will bound $\Pr[U_i]$. We assume $s<s'$. For computing $\Pr[\text{inc}^i(S^s)\oplus S^{s'} = 0],$\\
		In case of encryption of both message, where $P_1^s \neq P_{1}^{s'}.$  $CC^s$ is chosen randomly hence for any $i$
		\begin{equation}\label{31}
		\Pr[\text{inc}^i(S^s)\oplus S^{s'} = 0] = \frac{1}{2^n}.
		\end{equation}
		Similarly, in case of decryption of both the message where $MM^s$ is chosen randomly hence the same probability of (\ref{31}) holds.\\
		
		If one of them is decryption and other one is encryption. Without loss of generality, if $ty^s$ = \textit{dec} then both $S^s$ and $S^{s'}$ are polynomials of $h_1$ or $h_2$ of degree at most $\ell = \text{max}\{m^s + t^s,m^{s'} + t^{s'}\}$
		Then using Lemma \ref{lem:2}, we have
		\begin{equation}\label{111}
		\Pr\bigg[\bigvee_{i=0}^{m^s-2}U_i\bigg] \leq \dfrac{w_{max}{\ell}(m^s-1)}{2^n}.
		\end{equation}
		Similarly, we have
		\begin{equation}\label{222}
		\Pr\bigg[\bigvee_{i=0}^{m^s-2}V_i\bigg] \leq \dfrac{w_{max}{\ell}(m^s-1)}{2^n}.
		\end{equation}
		Thus, using equations (\ref{111}) and (\ref{222}), we have
		\begin{equation}\label{34}
		\Pr[{\sf COLL}(Z_s,Z_{s'})] \leq \dfrac{w_{max}(m^s+m^{s'}-2)\ell}{2^n}.
		\end{equation}
		Using equations (\ref{34}) and (\ref{27}) we have
		\begin{eqnarray}\nonumber
		\Pr[{\sf COLLD}_2] &\leq& \sum_{1 \leq s < s' \leq q}\frac{w_{max}(m^s+m^{s'}-2)\ell}{2^n}\\\nonumber
		&\leq& \frac{\ell w_{max}}{2^n}\sum_{1 \leq s < s' \leq q}(m^s+m^{s'})\\\nonumber
		&\leq& \frac{\ell w_{max}q \sigma_n}{2^n}.
		\end{eqnarray}
	\end{proof}
\end{claim}

\begin{claim}\label{claim:7}
	$\Pr[{\sf COLLD}_3] \leq \binom{q}{2}/2^n.$\\
	The proof is similar to the proof of the Claim \ref{claim:5}.
\end{claim}

\begin{claim}\label{claim:8}
	$\Pr[{\sf COLLR}_1] \leq \frac{(q-1)\sigma_n}{2^n}.$	
	\begin{proof}
		In case $P_1^s \neq P_1^{s'}$ and  $ty^s = ty^{s'} = enc$ then $CC^s$ and $CC^{s'}$ are selected uniformly at random from $\nbits$. Then, $\Pr[CC^s = CC^{s'}] = 1/2^n$.\\
		In case
		If $ty^s = dec$, then
		$$CC^s = MM^s \oplus H_{h_2}(R^s, T^s) \oplus H_{h_1}(Q^s, T^s),$$
		where
		$Q^s = P_2^s\|P_3^s\|\dotso\|P_{m^s}^s$ and $R^s = C_2^s\|C_3^s\|\dotso\|C_{m^s}^s$.\\
		Now, we have the following two cases to solve:\\
		
		\textbf{Case I:} When $ty^{s'} = enc$ then $CC^s$ is selected uniformly and independently from $\nbits$, then $\Pr[CC^s = CC^{s'}] \leq 1/2^n$.\\
		
		\textbf{Case II:} When $ty^{s'} = dec$ and $MM^s = MM^{s'}$. In this case we have
		$$ CC^s \oplus CC^{s'} =  \big[H_{h_2}(R^s, T^s) \oplus H_{h_2}(R^{s'}, T^{s'})\big] \oplus \big[H_{h_1}(Q^s, T^s) \oplus H_{h_1}(Q^{s'}, T^{s'})\big].$$
		Let
		
		$$H_2^{s,s'} = H_{h_2}(R^s, T^s) \oplus H_{h_2}(R^{s'}, T^{s'}),$$
		$$H_1^{s,s'} = H_{h_1}(Q^s, T^s) \oplus H_{h_1}(Q^{s'}, T^{s'}).$$
		
		Note that $H_1^{s,s'} \oplus H_2^{s,s'}$ is non-zero bivariabe polynomial on $h_1, h_2$ with (total) degree $\ell$. Hence, from Schwartz-Zippel Lemma
		$$\Pr[CC^s \oplus CC^{s'} = 0] \leq \frac{\ell}{2^n}.$$
		Therefore, we have
		\begin{eqnarray}
		\Pr[{\sf COLLR}_1] &\leq& \sum_{1 \leq s \leq s' \leq q} \frac{\ell}{2^n}\\\nonumber
		&\leq& \frac{(q-1)\sigma_n}{2^n}.
		\end{eqnarray}
	\end{proof}
\end{claim}

\begin{claim}\label{claim:9}
	$\Pr[{\sf COLLR}_2] \leq \binom{\sum_s m^s-q}{2}/2^n.$
	\begin{proof}
		From the game RAND2, we have for $1 \leq s \leq q,$ $Y_j^s =C_{j+2}^s \oplus P_{j+2}^s,$ where $0 \leq j \leq m^s-3,$ and 
		\[Y_{m^s-2}^s =
		\begin{cases}
		\text{{\sf pad}}(P_{m^s}^s) \oplus D_{m^s}^s & \quad \text{if}\; s=enc,\\
		\text{{\sf pad}}(C_{m^s}^s) \oplus V_{m^s}^s & \quad \text{if}\; s=dec.
		\end{cases}
		\]
		Hence, there are $\sum_{s} m^s -q$ uniformly and independently generated $n$-bit strings. Hence, the Claim follows.
	\end{proof}
\end{claim}

\begin{claim}\label{claim:10}
	$\Pr[{\sf COLLR}_3] \leq \frac{(q-1)\sigma_n}{2^n}.$\\
	Proof of the Claim is similar to the Claim \ref{claim:8}.
\end{claim}
Now by using Claim \ref{claim:5} to \ref{claim:10} and equation \ref{15}, we get :

\begin{eqnarray}
\nonumber
\textbf{Adv}^{\pm rnd}_{\text{XCBv1[3Perm}(n)]}(A) &\leq& \frac{2}{2^n}\binom{q}{2} + \frac{\ell q \sigma_n w_{max}}{2^n}+ \frac{1}{2^n}\binom{\sum_{s=1}^q m^s-q}{2} + \frac{2(q-1)\sigma_n}{2^n}\\
&\leq& \frac{2\sigma_n^2}{2^n} + \frac{\ell q \sigma_n w_{max}}{2^n}.\label{36}
\end{eqnarray}

By using above bound, equation (\ref{idea}) and the fact $w_{max} \leq 2^5$ as stated in \cite{Iwata}. We get the bound of XCBv1 as stated in Theorem \ref{imp_xcbv1} i.e.
$$\textbf{Adv}^{\pm \widetilde{prp}}_{\text{XCBv1[3Perm}(n)]}(A) \leq  \frac{(3+2^5)\ell q\sigma_n}{2^n}.$$

Similarly, we can prove the bound for XCBv2 as stated in Theorem \ref{imp_xcbv2}.


\section{Security of MXCB}
In Claim \ref{claim:6}, if we change the {\sf Ctr} mode used in XCBv1 to
$$\text{{\sf Ctr}}_{K,S}(P_1, \dotsc, P_{m}) = (P_1 \oplus E_K(S \oplus 1),\dotsc, P_{m} \oplus E_K(S \oplus m)),$$ 
where $S$ is the counter and $K$ is the key, let's say this is ``new counter mode''. Thus, we have a modified XCBv1, say, MXCBv1. Then MXCBv1's security bound will be different from the original due to ``new counter mode''. Proof of the security bound of this MXCBv1 will be similar to the previous, but we notice the change in the security bound due to ``new counter mode''. All the games are similar as we have shown, in the original XCBv1 we only replace Counter mode to 	``new Counter mode''. The collision probability of all the multi sets will be same except $\mathcal{D}_2$. Let us rename ${\sf COLLD}_2$ as ${\sf COLLD}_{\sf mod}$ which we calculate as follows:

\begin{claim}\label{claim:6a}
	$\Pr[{\sf COLLD}_{\sf mod}] \leq \binom{\sigma_n - q}{2}\frac{1}{2^n}.$
	\begin{proof}
		$\mathcal{D}_2 = Z_1 \cup Z_2 \cup \dotso \cup Z_q $, where\\
		\qquad$Z_s = \{S^s \oplus {\sf bin}_n(j) : 1 \leq j \leq m^s-1\},$ for $1 \leq s \leq q$
		and 
		\[S^s =
		\begin{cases}
		CC^s \oplus H_{h_1}(P_2^s \| \dotso \| P_{m^s}^s, T^s) & \text{if}\; ty^s = enc,\\
		MM^s \oplus H_{h_2}(C_2^s \| \dots \|C_{m^s}^s, T^s) & \text{if}\; ty^s = dec.
		\end{cases}
		\]
		It is easy to see for $x_1, x_2  \in Z_s$ then $\Pr[x_1=x_2] =0$. So our main task is to figure out $\Pr[x_1=x_2]$ when $x_1 \in Z_s$ and $x_2 \in Z_{s'}$, where $s \neq s'.$
		$$\Pr[ {\sf COLLD_{mod}}] \leq \sum_{1 \leq s < s' \leq q} \Pr[{\sf COLL}(Z_s, Z_{s'})].$$
		In case of $ty^s= ty^{s'}= enc$ and $P_1^s \neq P_1^{s'}$, $CC^s$ is chosen randomly and independently, for any $i$ and $j$
		$$\Pr[S_i^s = S_j^{s'}] = \dfrac{1}{2^n}.$$
		Similarly, for the case $ty^s= ty^{s'}= dec$ and $C_1^s \neq C_1^{s'}$, for any $i$ and $j$
		$$\Pr[S_i^s = S_j^{s'}] = \dfrac{1}{2^n}.$$
		In rest of the case where $ty^s = enc$ and $ty^{s'} = dec$,
		for $i^{th}$ and $j^{th}$ block, where $1 \leq i \leq m^s-1$ and $1 \leq j \leq m^{s'}-1,$\\
		$Z_s^i = Z_{s'}^j$ implies\\ $${\sf bin}_n{(i)} \oplus CC^s \oplus H_{h_1}(P_2^s \| \dotso \| P_{m^s}^s, T^s) ={\sf bin}_n{(j)} \oplus MM^{s'} \oplus H_{h_2}(C_2^{s'} \| \dots \|C_{m^{s'}}^{s'}, T^{s'})$$
		Let $s \leq s'$ and $(s,i) \neq (s',j)$. Thus, either $CC^s$ (in case $s^{th}$ is encryption query) or $MM^s$ (in case $s'^{th}$ is decryption query) is uniformly and independently distributed with all other variables. Thus, collision probability is $1/2^n$.\\
		Hence, for $\sigma_n-q$ messages 
		$$\Pr[{\sf COLLD}_{\sf mod}] \leq \binom{\sigma_n - q}{2}\frac{1}{2^n}.$$
	\end{proof}
\end{claim}
\par Now we can calculate the security bound for the MXCBv1 which is similar to the XCBv1. Difference in the security bound is due to the Claim \ref{claim:6} which we will replace by the Claim \ref{claim:6a}. Thus, Except the Claim \ref{claim:6}, by using the Claim \ref{claim:5} to \ref{claim:6a}, we get:
\begin{eqnarray}
\nonumber
\textbf{Adv}^{\pm rnd}_{\text{MXCBv1}[3\text{Perm}(n)]}(A) &\leq& \frac{2}{2^n}\binom{q}{2} + \frac{2}{2^n}\binom{\sum_{s=1}^q m^s-q}{2} + \frac{2(q-1)\sigma_n}{2^n}\\\nonumber
&\leq& \frac{q^2}{2^n} + \frac{({\sigma_n-q})^2}{2^n} + \frac{2\sigma_nq}{2^n}\\\nonumber
&\leq & \frac{2q^2+\sigma_n^2}{2^n}.
\end{eqnarray}

By using above bound and equation (\ref{idea}). We get the bound of MXCBv1 i.e.
$$\textbf{Adv}^{\pm \widetilde{prp}}_{\text{MXCBv1}[3\text{Perm}(n)]}(A) \leq  \frac{2.5 q^2+\sigma_n^2}{2^n}.$$

As we stated earlier, similarly we can state the security bound for the MXCBv2 (modified XCBv2). Thus, we can state the security bound of MXCBv1 and MXCBv2fb. For MXCBv1, security bound is as follows:
\begin{theorem}\label{imp_thm_xcbv1_mod}
	Consider an arbitrary adversary $A$ which queries only with messages/ciphers whose lengths are multiples of $n$ and $A$ asks a total of $q$ queries of overall query complexity $\sigma_n$ where each query is at most $\ell$ blocks long (each block of $n$ bits). Then,
	\begin{equation}\label{mxcbv1}
	{\bf Adv}^{\pm\nprp}_{\text{MXCBv1}[3\text{Perm}(n)]}(A) \leq \frac{2.5 q^2+\sigma_n^2}{2^n}.
	\end{equation}
\end{theorem}
Similarly, for  MXCBv2 our security bound is given as:
\begin{theorem}\label{mxcbv2}
	Consider an arbitrary adversary $A$ which queries only with messages/ciphers whose lengths are multiples of $n$ and $A$ asks a total of $q$ queries of overall query complexity $\sigma_n$ where each query is at most $\ell$ blocks long (each block of $n$ bits). Then,
	\begin{equation}\label{imp_xcbv2_mod}
	{\bf Adv}^{\pm\nprp}_{\text{MXCBv2fb}[3\text{Perm}(n)]}(A) \leq \frac{3.5 q^2+\sigma_n^2}{2^n}.
	\end{equation}
\end{theorem}

\section{Comparison of TES Security bounds}
In this section, we compare our derived security bound with the bound compared in \cite{chakraborty} with the same practical values of the parameter as taken in \cite{chakraborty}. So, Available ciphertext/plaintext to adversary is $2^{42}$. Also, the block length is 16 bytes. So, total ciphertext/plaintext is $2^{42}/2^4 = 2^{38}$ blocks. And sector size is 4 KB i.e. $2^{12}$ bytes. Thus, we take message length $2^{12}/2^4 = 2^8$ blocks. Therefore, we have total number of message is $2^{38}/2^8  = 2^{30}$. So, the total query complexity of the adversary is  $2^{38}+2^{30} = 2^{38.006}$, where $2^{30}$ is the tweak for adversary.

\par Hence, $q$ is number of queries i.e. $2^{30}$, $\ell$ is maximum query length i.e. $2^8+1$ and $\sigma_n$ query complexity i.e $2^{38.006}.$

\begin{table}\centering 
	\renewcommand*{\arraystretch}{2.5}
	{\rowcolors{3}{black!10!black!10}{white!70!white!40}
		\begin{tabular}{ |P{2.7cm}|P{3cm}|P{3cm}|P{3cm}|}
			
			\hline
			\multicolumn{4}{|c|}{\textbf{List of some TES}} \\
			\hline
			\centering{TES mode} & Source & Claimed Bound & Numerical Value\\
			\hline
			
			TET &	\cite{halevi2007invertible}	 
			& $\dfrac{3 \sigma_n^2}{2 \phi (2^n-1)}$ &  $2^{-50.40}$\\
			
			HCTR  & \cite{chakraborty2008improved} & $\dfrac{4.5 \sigma_n^2}{2^n}$ &  $2^{-49.81}$\\
			
			CMC   & \cite{halevi2003tweakable}   &	$\dfrac{7 \sigma_n^2}{2^n}$ &   $2^{-49.18}$\\
			
			EME   &  \cite{halevi2004parallelizable} & $\dfrac{7 \sigma_n^2}{2^n}$   & $2^{-49.18}$\\
			
			HEH, HMCH   &  \cite{sarkar2009efficient} & $\dfrac{20 \sigma_n^2}{2^n}$ & $2^{-47.66}$\\
			
			XCB   & \cite{mcgrew2007security} & $\dfrac{8 q^2 (\ell +2)^2}{2^n}$ & $2^{-48.96}$\\
			
			XCBv2fb	& \cite{chakraborty}  & 
			$\dfrac{(5+2^{22})\ell q \sigma_n}{2^n}$   &$2^{-29.98}$\\
			
			XCBv1   &  \cite{chakraborty}  			& $\dfrac{(5+2^{22})\ell q \sigma_n}{2^n}$   &$2^{-29.98}$\\
			
			Repaired XCBv2fb  & 	This Dissertation 					& 	$\dfrac{(5+2^{5})\ell q \sigma_n}{2^n}$ & $2^{-46.78}$\\
			Repaired XCBv1   & 	This Dissertation 					& 	$\dfrac{(3+2^{5})\ell q \sigma_n}{2^n}$ & $2^{-46.87}$\\
			
			MXCBv2fb 	& This Dissertation				& 	$\dfrac{3.5 q^2 + \sigma_n^2}{2^n}$ & $2^{-51.99}$\\
			
			MXCBv1  	& This Dissertation			& 	
			$\dfrac{2.5 q^2 + \sigma_n^2}{2^n}$&$2^{-51.99}$\\
			\hline
		\end{tabular}\caption{Comparison of the Bounds, here $\phi$ is Euler's totient.}
		\label{table:1}}
\end{table}

\par In the Table \ref{table:1}, we can see even repaired XCB gives the worst security bound as compare to listed TES scheme while MXCB gives the best bound.

\section{Weak keys analysis of XCB}	

In 2008, Handschuh and Preneel \cite{handschuh2008key} gave the following definition of weak keys:
\blockquote{
	In symmetric cryptology, a class of keys is called a weak key class if for the members of that class the algorithm behaves in an unexpected way and if it is easy to detect whether a particular unknown key belongs to this class. For a MAC algorithm, the unexpected behaviour can be that the forgery probability for this key is substantially larger than average. Moreover, if a weak key class is of size C, one requires that identifying that a key belongs to this class requires testing fewer than C keys by exhaustive search and fewer than C verification queries.}
\par According to this definition, a lot of work has been proposed by Handschuh and Preenel and Saarinen in MAC, AE based on polynomial hash function. Initially, Handschuh and Preenel considered 0 as the only weak key for all polynomial hash function. Later, in 2012 Marakku-Juhani \cite{saarinen2012cycling}, in 2015 Gorden Procter et al.\cite{procter2015weak} increased the number of weak keys for GCM. 

\subsection{Saarinen's Cycling attacks}
\par Saarinen's cycling attacks was proposed in 2012 in \cite{saarinen2012cycling} against GCM and other polynomial-based MACs and hashes. The main idea is, if a hash key $h$ lies in a subgroup of order $r$, then $h^r = 1$. Hence, for any $i, j$ message block $M_i$ and $M_{i+jr}$ can be swapped without changing the value of hash function. Saarinen also talk about the specific bit swapping instead of whole block swapping while the other condition of $M_i$ and $M_{i+jr}$ will remain same. The forgery technique is successful if the hash key is an element of low order subgroup with order dividing the distance between the swapped message space. This method identifies whether the hash key is in that class or not, by one valid message, tag pair and single verification query. Saarinen observes that any $r$ that divides $2^{128}-1$ can be used and that swapping of $M_i$ and $M_{i+r}$ will give successful forgery with probability at least $\dfrac{r+1}{2^{128}}$.  

\subsection{Saarinen's Cycling attacks on XCB}
\par As XCB contain hash function, a relevant question which arise about XCB - Is it contains weak keys? Some of weak keys attack on TES are present in \cite{sun2015weak}, here we will present how Saarinen's cycling attack (2012) which was done against GCM and hashes, can be done in XCBv1 and XCBv2. 
\par In XCBv2, line 108 in \autoref{fig : XCBv1_v2}, if a hash key $h$ is of order $t$, then $h^t=1$. Therefore, we can swap a plain text block $P_i$ and $P_{i+jt}$ (for some $i, j$) without changing the original value of the hash function.

$$108.\quad S \la CC \oplus H_{h}(0^n\|T, P_1\|\dotso\|P_{m-2}\|\text{{\sf pad}}(P_{m-1})\|0^n).$$

\par For example, if $h^3 =1$ then we can swap plaintext $P_i$ and $P_{i+3j}$, where $1\le i,i+3j \le m-1$ hence the plaintext $P$ will change but not the hash function value. Therefore, In line 109, counter mode will be same after swapping the plaintext block. Before swapping, let ciphertext $C$, plaintext $P$ and counter mode ${\sf Ctr}_{K_c,S}$ and after swapping ciphertext $C'$, plaintext $P'$ and counter mode will be same i.e. ${\sf Ctr}_{K_c,S}$. Now as we know the plaintext and ciphertext before swapping, therefore we can determine the Counter mode from line 109 i.e   $C \oplus P =  {\sf Ctr}_{K_c,S}$ (here $\oplus$ is block wise XOR, as per XCB scheme). As counter mode before and after swapping is same, so we can easily figure out the $C'$ i.e. $C' = P' \oplus {\sf Ctr}_{K_c,S}$. Thus, the cycling attack took place.

$$109.\quad (C_1, \dotsc, C_{m-1} )\la {\sf Ctr}_{K_c,S}(P_1, \dotsc, P_{m-2}, P_{m-1}).$$

\par Similarly, we can perform the weak keys attack on XCBv1.


\chapter{Security of HCTR}\label{chap : HCTR}

HCTR was proposed by Wang, Feng and Wu in 2005. It is a mode of operation which provides a tweakable strong pseudorandom permutation \cite{wang2005hctr}.

\par In this chapter,  we show how the hash function is insecure. We perform distinguishing attack and the hash key recovery attack on HCTR. Next we analyse the dependency of the two different keys in HCTR. In particular, we analyse the following scenario. Suppose HCTR with keys $K$ and $h$ has been used for some time, and $K$ gets compromised. We show that only changing $K$ would rise to a completely insecure scheme. 

\section{Description of HCTR}

\par HCTR comprises a block cipher $E$ and hash function $H$. It is a Tweakable Encipher Scheme with two master key : hash key $h$ and block cipher key $K$. HCTR and XCB have the similar structure i.e. these are hash-counter-hash Tweakable Enciphering Scheme. But HCTR's hash function and ${\sf Ctr}$ mode definition is different from the XCB.

Following is the definition of the HCTR as defined in \cite{wang2005hctr}:

\par \textbf{Polynomial hash function in HCTR :} Let $P$ be the plaintext such that $|P| = n(m-1)+r$ for $1 \leq r \leq n$. Partition of $P$ into $P_1||P_2||\dots||P_m$, where $|P_i| = n$ for $1 \leq i \leq m-1$, $0^{n-r}$ will append at the end of $P$ to complete the block such that all block size will be $n$.
$$ H_h(P) : \{0,1\}^n \times \{0,1\}^* \rightarrow \{0,1\}^n,$$
and the hash function $H$ is defined as
\[
H_h(P) = 
\begin{cases}
h, & \text{if $P$ is an empty string},\\
P_1h^{n+1} \oplus \dots \oplus (P_n||0^{n-r})h^{2}\oplus ((|P|)h), & \text{otherwise.}
\end{cases}
\]

\par \textbf{Counter mode in HCTR:} Given an $n$-bit string $S$, the counter mode {\sf Ctr} is defined as:
$\text{{\sf Ctr}}_{K,S}(A_1, \dotsc, A_{m}) = (A_1 \oplus E_K(S \oplus 1),\dotsc, A_m \oplus E_K(S \oplus m),$ where $S$ is the counter and $K$ is the key.

\par HCTR's encryption algorithm is shown in \autoref{fig:hctr} and construction in the \autoref{cons:hctr} with plaintext $|P|\ge n$ and tweak $|T| \ge 0.$\\

\begin{figure}[h]
	\begin{center}
		\begin{tabular}{|c|}
			\hline\\
			\begin{minipage}{370pt}
				Encryption under HCTR : $\textbf{E}_K^T(P)$\\\\
				1.  Partition $P$ into $P_1, \dots, P_m$\\
				2.  $CC \leftarrow P_1 \oplus H_h(P_2, \dots, P_m||T)$\\
				3.  $S \leftarrow CC \oplus E_K(CC) $\\
				4.  $C_2, \dots, C_m \leftarrow {\sf Ctr}_{K,S}(P_2, \dots, P_m)$\\
				5.  $C_1 \leftarrow E_K(CC) \oplus H_h(C_2, \dots, C_m ||T)$\\
				6.  \textbf{return} $(C_1, \dots, C_m)$\\
			\end{minipage}\\
			\hline
		\end{tabular}
	\end{center}
	\caption{\label{fig:hctr} Encryption using HCTR}
\end{figure}

\begin{figure}[t]
	\includegraphics[width=0.75\textwidth]{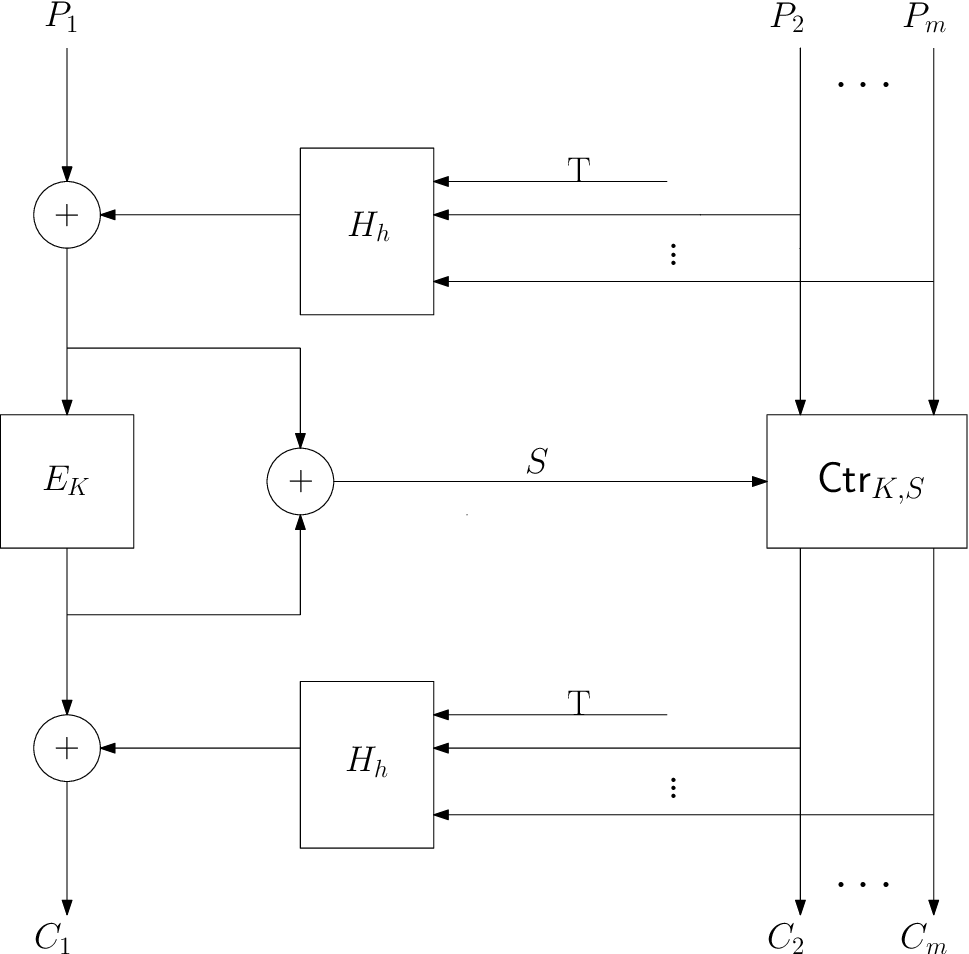}
	\centering
	\caption{\label{cons:hctr} Encryption of HCTR}
\end{figure}

\section{Insecurity of the hash function}
\par In 2005, the author showed a cubic security bound for HCTR in \cite{wang2005hctr}. Later in 2008, Chakraborty and Nandi gave the quadratic Security bound of HCTR \cite{chakraborty2008improved} which showed distinguishing advantage of an adversary in distinguishing HCTR and its inverse from a random permutation and its inverse is bounded above by $\frac{4.5 \sigma^2}{2^n}$ where $n$ the block-length of the block-cipher and $\sigma$ is the number of $n$-block queries made by the adversary (including the tweak). In this section, we show how the above claim is contradictory.  
\subsection{The Case of Empty Message}
\par Here, we give distinguish attack on the HCTR. In distinguish attack on the construction of HCTR, an adversary distinguish between the scheme and random permutation with high probability. Distinguish attack on HCTR which we show here generate the same internal value (i.e. $CC$) after performing hash function for two different messages. In HCTR's hash function if $P$ is empty string (as an input of hash function) or $P = 0$ (only single 0 as an input of hash function) then the $H_h(P)$ return the same counter value i.e. $h$, which is collision in hash function. Below we explain how distinguish attack can be performed on HCTR. 

\par Suppose an adversary makes two queries $P^{(1)}$ and $P^{(2)}$ such that $C^{(1)} = \text{HCTR}_{K,h}(P^{(1)})$ and $C^{(2)} = \text{HCTR}_{K,h}(P^{(2)})$ for empty tweak. Assume that $P^{(1)} = x$ and $P^{(2)} = (x\|0)$ where $x \in \nbits$ is arbitrary. Now,
\begin{enumerate}
	\item Due to collision in hash function for empty input and $0$ input, internal value $CC^{(1)}$ and $CC^{(2)}$ have the same value i.e. $x \oplus h.$ And so $E_K(CC^{(1)}) = E_K(CC^{(2)})$ , say, $E_K(CC)$.
	\item If $C_2^{(2)}$ has output $0$ then $C_1^{(1)} = C_1^{(2)}$ i.e. $E_K(CC) \oplus h$.
	\item Also, $\Pr[C_2^{(2)} = 0] = 1/2.$
\end{enumerate}
Thus, the advantage of the adversary is $(\frac{1}{2}-\frac{1}{2^n})$ which is very high. Therefore, the bound proved in \cite{chakraborty2008improved} is contradictory.

\par Not only distinguish attack, adversary can perform hash key recovery attack.  If we make the same set-up as above and $C_2^{(2)}$ is $1$ instead of $0$ then
\begin{eqnarray}\nonumber
C_1^{(1)} \oplus C_1^{(2)} &=& E_K(CC) \oplus h \oplus E_K(CC) \oplus h^2 \oplus h\\\nonumber
& = & h^2.
\end{eqnarray}
After getting $h^2$, adversary can easily retrieve the hash key $h$. Therefore, in $k$ iteration where $k \ge 1$, adversary have a high probability of retrieving hash key $h$ i.e.
$$\Pr\bigg[C_2^{(2)}=1\bigg] = \bigg(1 - \frac{1}{2^k}\bigg).$$ 

\par Thus, in the above discussion adversary not only perform the distinguishing attack but also retrieve the hash key $h$ with almost surety from HCTR scheme.

\subsection{Comment about the Attack}
In above two attacks, the weak point is definition of the hash function. Therefore, we can avoid the above attacks by putting some restriction either on the input query or hash function, or both. Here are the following way which can help us to avoid these attacks:
\begin{enumerate}
	\item We can avoid these attacks if we do not define the hash function for empty string and exclude  message of length $n$-bits or less from the message space of HCTR i.e. $|P|> n.$
	\item Above technique put restriction on the message. We can prevent the above attack without changing the message size of HCTR just by modifying the definition of hash function i.e increase the size of the input by one. So, the new hash $H^{'}$ is defined as: $H_h^{'}(P) = H_h(P||1).$
\end{enumerate}

\section {Key Dependency in HCTR}

\par In definition of HCTR, we didn't talk about the two master keys whether they have some relation or not. Suppose in an implementation, Two master keys $h_1$ (hash key) and $K_1$ (block cipher key) were used and encryption, decryption took place. Suppose after some time, block key $K_1$ gets compromised and it changed from  $K_1$ to $K_2$. Also, we have sufficient amount of data available corresponding to the old key pair. Here, we show how one can extract the hash key i.e. $h_1$ with certainty and the whole scheme is compromised i.e. hash key recovery attack is possible. Thus, the scheme will not be secure any-more. 

\par Suppose an adversary has a plaintext ciphertext pair $(P, C)$ such that $C = \text{HCTR}_{K_1,h_1}(P)$ for empty tweak. Assume that $P = x\|x$ where $x \in \nbits$ is arbitrary. Now,
\begin{equation}\label{hctr:kd1}
C_1 = E_{K_1}(CC) \oplus H_{h_1}(C_2),
\end{equation}

where $CC = x \oplus H_{h_1}(x).$  

$$C_2 = E_{K_1}(S \oplus {\sf bin}_n{(1)}) \oplus x,$$ 
Here,
\begin{eqnarray}\nonumber
S  &=& CC \oplus E_{K_1}(CC)\\
&=& x \oplus H_{h_1}(x) \oplus E_{K_1}(CC).\label{hctr:kd2}
\end{eqnarray}

Now, as we know $C_2, x$ and $K_1$, we have 
\begin{equation}\label{hctr:kd3}
S = E_{K_1}^{-1}(C_2 \oplus x) \oplus {\sf bin}_n(1).
\end{equation}

By using equations (\ref{hctr:kd1}), (\ref{hctr:kd2}) and \ref{hctr:kd3}, we get
\begin{eqnarray}\nonumber
C_1 \oplus S \oplus x &=& E_{K_1}(CC) \oplus H_{h_1}(C_2) \oplus x \oplus H_{h_1}(x) \oplus E_{K_1}(CC) \oplus x,\\
&=& (x \oplus C_2){h_1}^2.\label{hctr:kd4}
\end{eqnarray}

Here equation (\ref{hctr:kd4}) is a quadratic equation in $h_1$ which we can solve easily and retrieve the hash key $h_1$.  

\par To prevent from this attack, we should change both the keys simultaneously. Note we did not comment on changing only the hash key and keeping block cipher key as it is. 

\chapter{Conclusion}\label{chap : Conclusion}
	In this dissertation, we improved the security bound of a Tweakable Enciphering Scheme (TES) known as XCB (Extended Code Book). Also, we proposed a Modified XCB (MXCB) and showed that the security bound of MXCB has better numerical value than many other popular TES like HCTR, HCH, TET, HEH, CMC, XCB and EME. We also analysed some weak keys attack on XCB.
	 
	\par Further, we analysed a TES known as HCTR. We performed distinguishing and key recovery attack on the existing HCTR and also showed how it can be avoided easily. We also showed why both the master keys of HCTR should be changed simultaneously or otherwise it could be a serious attack on the construction of HCTR.

\bibliographystyle{plain}

\bibliography{bibaana}
\end{document}